\documentclass[reqno,10pt]{amsart}

\usepackage[left=3cm,right=3cm,top=3cm,bottom=3cm]{geometry}
\usepackage{amssymb,bm}
\usepackage{mathrsfs}
\usepackage{xcolor}
\usepackage{enumitem}
\usepackage{threeparttable}
\usepackage{hyperref}

\usepackage{graphicx}
\usepackage{subfig}

\newtheorem{theorem}{Theorem}[section]
\newtheorem{lemma}[theorem]{Lemma}
\newtheorem{proposition}{Proposition}[section]

\theoremstyle{definition}

\theoremstyle{remark}
\newtheorem{remark}[theorem]{Remark}

\numberwithin{equation}{section}

\begin{document}

\title[Poisson-EQMOM]{Poisson quadrature method of moments for 2D kinetic equations with velocity of constant magnitude}


\author{Yihong Chen}
\address{Department of Mathematical Sciences, Tsinghua University, Beijing 100084, China}
\curraddr{}
\email{chenyiho20@mails.tsinghua.edu.cn}
\thanks{}

\author[Qian Huang]{Qian Huang*}
\address{Department of Energy and Power Engineering, Tsinghua University, Beijing 100084, China}
\curraddr{}
\email{huangqian@tsinghua.edu.cn; hqqh91@qq.com}
\thanks{* Corresponding author}

\author{Wen-An Yong}
\address{Department of Mathematical Sciences, Tsinghua University, Beijing 100084, China \\ Yanqi Lake Beijing Institute of Mathematical Sciences and Applications, Beijing 101408, China}
\curraddr{}
\email{wayong@tsinghua.edu.cn}
\thanks{This work is supported by National Key Research and Development Program of China (Grant no. 2021YFA0719200) and National Natural Science Foundation of China (Grant no. 12071246).}

\author{Ruixi Zhang}
\address{Department of Energy and Power Engineering, Tsinghua University, Beijing 100084, China}
\curraddr{}
\email{1553548358@qq.com}
\thanks{}

\subjclass[2020]{Primary 35F50 · 35Q82 · 82-10}

\keywords{kinetic equation, moment method, Poisson kernel, active matter, Vicsek model}

\date{}

\dedicatory{}

\begin{abstract}

This work is concerned with kinetic equations with velocity of constant magnitude. We propose a quadrature method of moments based on the Poisson kernel, called Poisson-EQMOM.
The derived moment closure systems are well defined for all physically relevant moments and the resultant approximations of the distribution function converge as the number of moments goes to infinity.
The convergence makes our method stand out from most existing moment methods. Moreover, we devise a delicate moment inversion algorithm.
As an application, the Vicsek model is studied for overdamped active particles, and the Poisson-EQMOM is validated with a series of numerical tests including spatially homogeneous, one-dimensional and two-dimensional problems.
\end{abstract}

\maketitle

\section{Introduction}

In this work we are interested in kinetic equations with velocity of constant magnitude (KE-VC).
Such equations can be used to describe the motion of overdamped active particles (e.g., cells) with self-propulsive forces of a constant magnitude \cite{bertin2006,Borsche2019,degond2008,ihle2011,marchetti}.
Among KE-VC are also the radiative transfer equations of photons and neutrons \cite{chanrte,KH_Moment13,mihalas}.
Besides, such kinetic equations arise in the `semi-continuous' discretization of the Boltzmann equation with the velocity space replaced by a series of spheres \cite{PRsemi}.

Like the usual kinetic equations, the KE-VC is computationally costly due to the high dimensionality.
Thus, the method of moments as a model reduction technique attracts much attention \cite{Picard23}.
This method provides a successful bridge connecting the kinetic equation to the hydrodynamic theory.
The resultant systems have clear physical interpretations for the lower-order moments, are computationally more feasible, can be valid over a wider range of flow regimes, and naturally inherit the conservation laws of the kinetic equation.
The application areas include rarefied gas dynamics \cite{dvk21,Grad1949}, plasma physics \cite{Tau23}, multiphase and/or particulate-laden flows \cite{MarFox2013}, radiative transfer \cite{KH_Moment13}, traffic flows \cite{Marq13} and active matter physics \cite{bertin2006,degond2008,ihle2011}.
For the general kinetic equations (with non-constant velocity magnitude), the Grad's closure \cite{Grad1949} is regarded as the first class of modeling approach, and recent years see remarkable progress in its hyperbolicity regularization (see e.g. \cite{Cai13, Picard23} and references therein). Another moment closure strategy is the quadrature-based method \cite{Mc1997,MarFox2013,Huang2020,zhang2023}, which ensures \textit{positivity} of the reconstructed distribution once the moments are realizable.

The existing moment methods for KE-VC mainly include the classical $P_N$ method (using the spherical harmonics as the basis) \cite{pomraning2005} and the maximum entropy (or termed $M_N$) method \cite{minerbo1978}.
However, the $P_N$ method does not guarantee the positivity of the reconstructed distribution and the resultant moment system may lead to unphysical results \cite{brunner2002}, whereas the numerical treatment of the $M_N$-derived systems turns out to be extremely difficult \cite{hauck2011}.
We further refer to \cite{KH_Moment13,Laiu2016} for filtered $P_N$ that is (almost) positivity preserving and to \cite{All2019,McD2013,SchTor2015} for some regularized $M_N$-based methods. Considering the state of the art, developing effective moment methods remains an active research direction.

This work is an attempt in this direction and aims to develop quadrature-based moment methods for KE-VC that are free of those drawbacks. It is inspired by the extended quadrature method of moments (EQMOM) originally proposed for the Boltzmann equation \cite{Chalons2010}.
In EQMOM, the distribution is approximated as a sum of several homoscedastic kernels with independent weights and centers. It thus preserves positivity of the distribution, and the parameters (weights, centers and variance) can be efficiently computed \cite{Pigou2018,MarFox2013}.
For the one-dimensional velocity, the Gaussian distribution is most-widely used as the kernel \cite{Chalons2017}.
The EQMOM has also become a popular method to simulate the evolution of aerosol size and compositional distributions using the lognormal or Gamma distribution (supported on $\mathbb R_+$) and the compactly-supported beta distribution as kernels, respectively \cite{Pigou2018,Huang2020k,Yuan2011}.
Further discussions and multidimensional extensions of EQMOM can be found in \cite{Chalons2017,Chen2024,Huang2020,Huangjsc,zhang2023}.

However, to our best knowledge, the combinations of EQMOM with KE-VC are only reported for the one-dimensional radiative transfer equations where the distribution is supported on a closed interval \cite{ALL2016,vikas2013}.
It remains unclear if (and how) an EQMOM method can be established for angular distributions supported on the sphere. This would require new techniques for function approximation on the sphere and nonlinear inversion algorithms to compute the parameters.
As a first step to tackle the issue, we propose that, for the 2D KE-VC, the distributions supported on the unit circle can be approximated by using the Poisson kernel.
This seems the first time to pick the Poisson kernel as a building block in the moment method.
Previously, the Poisson kernel was adopted as a density function of spherical distributions in the statistical mixture model \cite{yang2004}, aiming to infer a probability distribution from observed data using the maximum likelihood estimator. However, this approach is not relevant for solving kinetic equations.
The Poisson kernel was also involved in a discretized spectral approximation method in the one-dimensional neutron transport theory \cite{Seng1988}, which is different from our attempt.
As such, our method, denoted Poisson-EQMOM, is a novel way to formulate hydrodynamic theories of, e.g., planar flows of overdamped active matter (like cell motions).

The Poisson-EQMOM naturally inherits from EQMOM the advantages of \textit{being positivity-preserving} (which cannot be satisfied by the $P_N$ method) and \textit{having a conservative form}. Surprisingly, we show in Theorem \ref{thm:mainexst} that the derived moment closure systems are well defined for \textit{all physically relevant moments}, while the existing EQMOM variants may not have such a global property \cite{Chalons2017}.
More importantly, the resultant approximations of the distribution function \textit{converge} as the number of moments goes to infinity (see Theorem \ref{thm:lift}). This convergence is desired for any moment method but has not been demonstrated for other quadrature-based method of moments.
These properties indicate that the Poisson kernel is a very right choice for the EQMOM to handle KE-VC.
On the basis of these nice properties, we devise an efficient and robust moment inversion algorithm while the existing ones \cite{Pigou2018,MarFox2013} seem not to work in the current situation.

As an application, we apply the Poisson-EQMOM to a 2D KE-VC derived from the Vicsek dynamics \cite{degond2008} and obtain a new moment closure system. Let us remark that this KE-VC model was analyzed in \cite{Figalli2018} and numerically treated in \cite{filbet2018,GAMBA2015,Griette2019} with various schemes.
The EQMOM-derived system is solved numerically with a detailed calculation of the kinetic-based flux.
The results for five different cases, ranging from spatially homogeneous to 2D, are all quite satisfactory in comparison with analytical results, macroscopic solutions and microscopic particle simulations \cite{GAMBA2015}.

The remainder of the paper is organized as follows. Section~\ref{sec:method} presents the Poisson-EQMOM with its realizability and convergence properties, one of which is proved in Appendix A.
The moment inversion algorithm is developed in Section~\ref{sec:inv}.
Section~\ref{sec:app} is devoted to the application of the Poisson-EQMOM to a KE-VC from the Vicsek model.
A numerical scheme is given in Section~\ref{sec:sch}, whereas the numerical results are reported in Section~\ref{sec:num}. Conclusions are drawn in Section~\ref{sec:con}.

\section{Moment method with Poisson kernel} \label{sec:method}
Consider the 2D KE-VC for an angular distribution $f=f(t,\bm x, \theta)$ with $t>0$, $\bm x \in \mathbb R^2$ and $\theta \in \mathbb R$ (modulo $2\pi$):
\begin{equation} \label{eq:feq}
  \partial_t f + v_0\bm e_{\theta}\cdot \nabla_{\bm x} f = Q(f).
\end{equation}
Here $v_0$ is a constant speed, $\bm e_{\theta} = (\cos \theta, \sin \theta) \in \mathbb S^1$, and $Q(f)$ is a problem-specific collision operator. A typical example for the polar active flow will be detailed in Section \ref{sec:app}.

Our goal is to develop a moment method for solving the above kinetic equations. For this purpose, we define the $k$th angular moments of $f(t,\bm x,\theta)$ as
\begin{equation} \label{eq:csk}
  c_k = c_k(t,\bm x) = \int_{-\pi}^{\pi} f(t,\bm x,\theta) \cos k\theta d\theta, \quad
  s_k = s_k(t,\bm x) = \int_{-\pi}^{\pi} f(t,\bm x,\theta) \sin k\theta d\theta
\end{equation}
for $k\in \mathbb N$.
Thanks to the isomorphism between $\mathbb R^2$ and $\mathbb C$, it is convenient to write $m_k = c_k + i s_k \in \mathbb C$ with $i^2 = -1$.
We then see from Eq.(\ref{eq:csk}) that
\begin{equation}
  m_k = m_k(t,\bm x) = \int_{-\pi}^{\pi} f(t,\bm x,\theta) e^{ik\theta}d\theta
\end{equation}
and hence $m_{-k}=\overline{m_k}$ (the overline denotes the complex conjugate of a complex number).
In this way, the governing equation for $m_k$ can be deduced from Eq.(\ref{eq:feq}) as
\begin{equation} \label{eq:mkeq_unc}
  \partial_t m_k + v_0\left( \frac{1}{2}\partial_x + \frac{1}{2i}\partial_y \right) m_{k+1} + v_0\left( \frac{1}{2}\partial_x - \frac{1}{2i}\partial_y \right) m_{k-1} = \int_{-\pi}^{\pi} Q(f) e^{ik\theta} d\theta,
\end{equation}
where the identities $\cos \theta = (e^{i\theta}+e^{-i\theta})/2$, $\sin\theta=(e^{i\theta}-e^{-i\theta})/2$ have been used.
For the above equations in Eq.(\ref{eq:mkeq_unc}), a finite truncation on $k \le N$ leaves (at least) the highest-order moment $m_{N+1}$ being unclosed. This calls for a moment closure approach.

A classical approach is the $P_N$ method which assumes the distribution $f$ to be a truncated Fourier series and gives $m_{N+1}=0$. This is a convenient closure as higher-order Fourier coefficients are small.
Unfortunately, this approach does not guarantee the positivity of the truncated distribution and the resultant moment system may lead to unphysical results \cite{brunner2002}.
Another well-known approach is the maximum entropy method which can generate well-behaved systems \cite{minerbo1978}. However, the resultant numerical treatment turns out to be extremely difficult \cite{hauck2011}.
In view of such a state of the art, it remains desirable to develop a reasonable closure for the 2D KE-VC (\ref{eq:feq}) that is positivity-preserving and numerically efficient.
This is the very goal of this work.

\subsection{Poisson-EQMOM}

We now propose a new extended quadrature method of moments equipped with the Poisson kernel (denoted Poisson-EQMOM) for Eq.(\ref{eq:mkeq_unc}). The key idea is to approximate the angular distribution $f(\theta)$ with a convex combination of Poisson kernels:
\begin{equation} \label{eq:peans}
  f_N(\theta) = \sum_{\alpha=1}^N \rho_{\alpha} P_r(\phi_{\alpha} - \theta)
\end{equation}
that matches the $(N+1)$ lower-order moments of $f(\theta)$:
\begin{equation} \label{eq:sol}
  m_k = \int_{-\pi}^\pi f_N(\theta) e^{ik\theta} d\theta = \sum_{\alpha=1}^N \rho_{\alpha} \int_{-\pi}^\pi e^{ik\theta}P_r(\phi_\alpha-\theta) d\theta
\end{equation}
for $k=0,1,\dots,N$.

Obviously, such an ansatz introduces $(2N+1)$ unknowns collected as
\begin{equation}
  W=(\rho_1,\dots,\rho_N,\phi_1,\dots,\phi_N,r)
\end{equation}
to be solved by (\ref{eq:sol}) from the transported moments $m_0,\dots,m_N$ at each time $t$ and position $\bm x$.

The Poisson kernel in (\ref{eq:peans}) is defined for $\theta \in [-\pi,\pi)$ and $0<r<1$ as
\begin{equation} \label{eq:poisdef}
  P_r(\theta) = \frac{1}{2\pi}\frac{1-r^2}{1-2r\cos\theta+r^2} = \frac{1}{2\pi} \sum_{k\in\mathbb Z} r^{|k|}e^{ik\theta},
\end{equation}
where the second expression gives the Fourier series of the Poisson kernel. It is seen in Fig.~\ref{fig:pois} that $P_r(\theta)$ converges to a Dirac delta function centered at $\theta=0$ as $r\to 1$.

\begin{figure}[htbp]
\centering
\includegraphics[height=4.5cm]{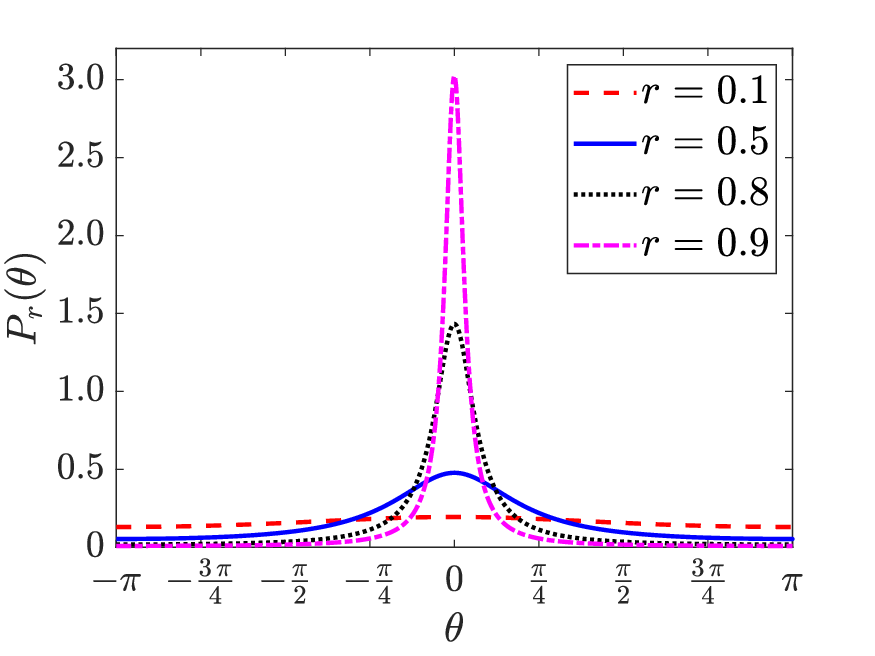}
\caption{The Poisson kernels with different values of $r$.}
\label{fig:pois}
\end{figure}

The Poisson kernel has the following elegant property, which is the consequence of Theorems 5 and 15 in Section 2.2 of \cite{evans}.

\begin{proposition} \label{prop:pois}
  For any complex-valued function $h(z)$ that is harmonic on $\{z\in\mathbb C: \ |z|<1\}$ and continuous on $\{z\in\mathbb C: \ |z|\le 1\}$, we have
  \[
    h(re^{i\phi}) = \int_{-\pi}^{\pi} h(e^{i\theta}) P_r(\phi-\theta) d\theta.
  \]
\end{proposition}
In particular, this implies
\begin{equation} \label{eq:pois}
  r^k e^{ik\phi} = \int_{-\pi}^{\pi} e^{ik\theta} P_r(\phi-\theta)d\theta
\end{equation}
for $k=0,1,2,\dots$. Thus, Eq.(\ref{eq:sol}) can be reformulated as
\begin{equation} \label{eq:poeqrm}
  m_k = r^k \sum_{\alpha=1}^N \rho_{\alpha} e^{ik\phi_{\alpha}}, \quad k=0,\dots,N,
\end{equation}
or equivalently,
\[
\begin{aligned}
  c_k &= r^k \sum_{\alpha=1}^N \rho_{\alpha} \cos k\phi_{\alpha}, \quad k=0,1,\dots,N, \\
  s_k &= r^k \sum_{\alpha=1}^N \rho_{\alpha} \sin k\phi_{\alpha}, \quad k=1,\dots,N.
\end{aligned}
\]

If $W$ is solved from Eq.(\ref{eq:poeqrm}), then all the unclosed terms in Eq.(\ref{eq:mkeq_unc}) can be evaluated according to the Ansatz (\ref{eq:peans}).
For instance, the highest-order moment $m_{N+1}$ is reconstructed as
\begin{equation} \label{eq:mnp1}
  \hat m_{N+1} = r^{N+1} \sum_{\alpha=1}^N \rho_{\alpha} e^{i(N+1)\phi_{\alpha}}.
\end{equation}
As a consequence, we derive a closed moment system for $\bm M_N=(m_0,\dots,m_N)^T \in \mathbb C^{N+1}$ which reads as
\begin{equation} \label{eq:momsys}
  \partial_t \bm M_N + \frac{v_0}{2}\partial_x \left(F_1(\bm M_N)+F_2(\bm M_N) \right) + \frac{v_0}{2i} \partial_y \left(F_2(\bm M_N)-F_1(\bm M_N) \right) = S(\bm M_N)
\end{equation}
with
\[
  F_1(\bm M_N) = (\overline{m_1}, m_0, \dots, m_{N-1} )^T, \quad
  F_2(\bm M_N) = (m_1, \dots, m_N, \hat m_{N+1})^T,
\]
and a model-specific source term $S(\bm M_N) \in \mathbb C^{N+1}$.

Clearly, solving the unknowns $W$ from Eq.(\ref{eq:poeqrm}) plays a vital role in the Poisson-EQMOM.
In the next subsection we show that such a $W$ exists and the resultant approximation $f_N=f_N(\theta)$ converges to $f$ as $N$ approaches infinity. This subsection is closed with the remark on possible extensions of our model.

\begin{remark}
  The Poisson-EQMOM may be extended to handle kinetic equations with non-constant speeds by discretizing the velocity space $\mathbb R^2$ into a series of concentric 1-spheres. Note that such discretization of the velocity space was studied in \cite{PRsemi} (called the `semi-continuous' method therein) for the Boltzmann equation.
  Furthermore, it would be important to develop analogous models in three dimensions, though directly assuming the distribution to be a convex combination of Poisson kernels on $\mathbb S^n$ ($n\ge 2$) could lead to difficulties: First, the number of unknowns is generally not the same as that of the moments; Second, the realizability of moments on $\mathbb S^n$ is incomplete for $n \ge 2$ (see an exposition in \cite{gtm277}). Detailed investigations are left for future work.
\end{remark}

\subsection{Existence and convergence of $f_N$}
Denote $\mathbb T=\{ z\in \mathbb C: \ |z|=1 \}$ as the unit circle.
Our main results in this subsection are stated as follows.
\begin{theorem} \label{thm:mainexst}
  Assume that $\bm M_N \in \mathbb C^{N+1}$ is realizable with respect to a certain measure $\nu$ on $\mathbb T$ (that is, $m_k=\int_{\mathbb T} z^k d\nu(z)$ for $k=0,\dots,N$).
  Then there exists a unique $r\in[0,1]$ and the corresponding $\rho_{\alpha}>0$, $\phi_{\alpha}\in [-\pi,\pi)$ ($\alpha=1,\dots,N$) solving Eq.(\ref{eq:poeqrm}), and any closure by the Poisson-EQMOM is well-defined.
\end{theorem}

\begin{theorem} \label{thm:lift}
  Assume that $f(\theta) \ge c>0$. The following statements are true.

  (i). If $f$ is absolutely continuous and $f':=df/d\theta \in L^2(-\pi,\pi)$, then we have $f_N \to f$ in $L^2(-\pi,\pi)$ as $N\to\infty$.

  (ii). If $f$ is Lipschitz, then we have $f_N \to f$ in $L^p(-\pi,\pi)$ for $2 \le p < \infty$.

  (iii). If $f$ is Lipschitz and $f'$ is of bounded variation, then we have $f_N \to f$ uniformly.
\end{theorem}

\begin{remark}
  Theorem \ref{thm:mainexst} reveals that the Poisson-EQMOM enjoys a global property that all physically relevant moments are solvable.
  By contrast, for the 1-D velocity distribution, the Gaussian-EQMOM cannot be defined for certain moments generated with reasonable distributions (see Proposition 3.1 in \cite{Chalons2017}). It is also worth mentioning that for the $M_N$ (maximum entropy) method \cite{minerbo1978,Levermore1996}, the realizability domain has a clear physical interpretation and even its singularity (along the so-called Junk line \cite{junk1998}) is linked to a feature, rather than an insufficiency. This would cause trouble in numerical computation, while our Poisson-EQMOM is free of such a problem.
\end{remark}

\begin{remark}
  Theorem \ref{thm:lift} demonstrates that the reconstructed distribution does converge to the original one if more moments are used, as long as the distribution has a positive lower bound.
  This property is definitely desired for any moment method, but seems not be verified for other kinds of quadrature-based method of moments.
  Let us also mention the standard $P_N$ method with the reconstructed $f_N(\theta)=(2\pi)^{-1}\sum_{k=-N}^N m_k e^{ik\theta}$ being the truncated Fourier series. It has `stronger' convergence results and does not require $f\ge c>0$. In particular, we have (i) $f\in L^2(-\pi,\pi) \Rightarrow f_N \to f$ in $L^2(-\pi,\pi)$ and (ii) $f\in C^\gamma (\mathbb T) \Rightarrow f_N \to f$ in $L^p(-\pi,\pi)$ with $2<p\le \infty$ and $\frac{p-2}{2p}<\gamma<1$. However, unlike our Poisson-EQMOM, the $P_N$ method cannot guarantee $f_N\ge 0$.
\end{remark}

The proof of Theorem \ref{thm:lift} is quite technical and is therefore presented in Appendix A.
In what follows we proceed to prove Theorem \ref{thm:mainexst}.

The main tool in our analysis is the Toeplitz matrix defined as \cite{gtm277}:
\begin{equation}\label{eq:2.13}
  H_k = H_k(\bm M_N) =
  \begin{bmatrix}
    m_0 & m_1 & \cdots & m_k \\
    \overline{m_1} & m_0 & \cdots & m_{k-1} \\
    \vdots & \vdots & \ddots & \vdots \\
    \overline{m_k} & \overline{m_{k-1}} & \cdots & m_0
  \end{bmatrix}
  \in \mathbb C^{(k+1)\times (k+1)}
\end{equation}
for $k=0, 1, \dots, N$. The $H_k$'s are Hermitian matrices with real eigenvalues.
We quote the following result, which is a consequence of Theorem 11.5, Proposition 11.15 and Proposition 11.16 of \cite{gtm277}.
\begin{lemma} \label{thm:toep}
  Given $\bm M_N\in\mathbb C^{N+1}$, the following two statements are equivalent.

  (i). The Toeplitz matrix $H_N(\bm M_N)$ is positive semi-definite and its rank $n$ is no larger than $N$.

  (ii). There exists an $n$-atomic measure $\mu = \sum_{\alpha=1}^n \rho_{\alpha} \delta_{z_{\alpha}}$ on $\mathbb T$ with $\rho_{\alpha}>0$ for all $\alpha$ such that $m_k = \int_{\mathbb T} z^k d\mu(z)$ for $k=0, \dots, N$. Such a measure $\mu$ is unique.
\end{lemma}

As a first step, for $r\in(0,1]$, Eq.(\ref{eq:poeqrm}) can be rewritten as
\[
  m_k r^{-k} = \int_{\mathbb T} z^k d\mu
\]
with $\mu = \sum_{\alpha=1}^N \rho_{\alpha} \delta_{z_{\alpha}}$ and $z_{\alpha}=e^{i\phi_\alpha}\in\mathbb T$ for $k=0,\dots,N$.
According to Lemma \ref{thm:toep}, this requires the Toeplitz matrix corresponding to
\begin{equation} \label{eq:MNstar}
  \bm M_N^*(r)=(m_0, m_1 r^{-1},\dots,m_N r^{-N}) \in \mathbb C^{N+1},
\end{equation}
denoted by $H_N(\bm M_N^*(r))$, to be positive semi-definite with $\det H_N(\bm M_N^*(r)) = 0$.
In other words, the minimum eigenvalue of $H_N(\bm M_N^*(r))$ must be zero.

Denote by $\lambda(r;N)$ the minimum eigenvalue of $H_N(\bm M_N^*(r))$, which is a continuous function on $r$. We further show that
\begin{proposition} \label{prop:eiginc}
  $\lambda(r;N)$ is strictly increasing on $r\in(0,1]$.
\end{proposition}

\begin{proof}
  For any $0<r_1<r\le 1$, we denote $\nu=\lambda(r_1;N)$ and $\bm{\tilde M}_N = (\tilde m_0,\dots,\tilde m_N)\in \mathbb C^{N+1}$ with
  \[
    \tilde m_0 = m_0 -\nu \ge 0 \text{    and  } \tilde m_k = m_k \quad \text{for }k=1,\dots,N.
  \]
  It is straightforward to verify that
  \[
    H_N(\bm{\tilde M}_N^*(r_1)) = H_N(\bm M_N^*(r_1)) - \nu I_{N+1}
  \]
  with $I_{N+1}$ being the unit matrix of order $N+1$, and hence the Toeplitz matrix $H_N(\bm{\tilde M}_N^*(r_1))$ is positive semi-definite with zero determinant.
  According to Lemma \ref{thm:toep}, there exits a unique $n$-atomic measure $\mu = \sum_{\alpha=1}^n \rho_{\alpha} \delta_{z_{\alpha}}$ for some $n\le N$ such that
  \[
    \tilde m_k r_1^{-k} = \int_{\mathbb T} z^k d\mu(z) = \sum_{\alpha=1}^n \rho_{\alpha} z_{\alpha}^k
  \]
  holds for $k=0,\dots,N$.
  Writing $z_{\alpha} = e^{i\phi_{\alpha}}$, we obtain
  \begin{equation} \label{eq:mkrmk}
    \tilde m_k r^{-k} = \left( r_1/r \right)^k \sum_{\alpha=1}^n \rho_{\alpha} e^{ik\phi_{\alpha}}
    = \sum_{\alpha=1}^n \rho_{\alpha} \int_{-\pi}^{\pi} e^{ik\theta} P_{r_1/r}(\phi_{\alpha}-\theta)d\theta,
  \end{equation}
  where the second equality results from Eq.(\ref{eq:pois}).

  Now we claim that the Toeplitz matrix $H_N(\bm{\tilde M}_N^*(r))$ is positive definite. To see this, take any nonzero $a=(a_0,\dots,a_N)^T\in\mathbb C^{N+1}$ and compute
  \[
  \begin{split}
    \bar a^T H_N(\bm{\tilde M}_N^*(r)) a
    &= \sum_{k,l=0}^N \overline{a_k} a_l \tilde m_{l-k} r^{-|l-k|} \\
    &= \sum_{\alpha=1}^n \rho_{\alpha} \int_{-\pi}^{\pi} \left[ \sum_{k,l=0}^N \overline{a_k}a_l e^{i(l-k)\theta} \right] P_{r_1/r}(\phi_{\alpha} - \theta) d\theta \\
    &= \sum_{\alpha=1}^n \rho_{\alpha} \int_{-\pi}^{\pi} \left | \sum_{l=0}^N a_l e^{il\theta} \right |^2 \ P_{r_1/r}(\phi_{\alpha} - \theta) d\theta
    > 0.
  \end{split}
  \]
  Here we have used Eq.(\ref{eq:mkrmk}) to obtain the second equality. Note that for $l<k$, Eq.(\ref{eq:mkrmk}) should be used with the complex conjugate taken on the both sides.

  Therefore, the minimum eigenvalue of $H_N(\bm{\tilde M}_N^*(r))$ is strictly positive.
  Further noticing
  \[
    H_N(\bm{\tilde M}_N^*(r))= H_N(\bm M_N^*(r)) - \nu I_{N+1},
  \]
  we conclude $\lambda(r;N) > \nu = \lambda(r_1;N)$, which completes the proof.
\end{proof}

The following fact will also be needed.
\begin{proposition} \label{prop:pd1}
  Suppose that $\nu$ is a measure on $\mathbb T$ with at least $(N+1)$ points of support. If $m_k=\int z^k d\nu$ for $k=0,\dots,N$, then the Toeplitz matrix $H_N(\bm M_N^*(1))=H_N(\bm M_N)$ is positive definite
\end{proposition}

\begin{proof}
  For any nonzero $a=(a_0,\dots,a_N)^T\in\mathbb C^{N+1}$, we define a polynomial $p(z)=\sum_{j=0}^N a_j z^j$ and compute
  \[
    \bar a^T H_N(\bm M_N^*(1)) a = \sum_{j,k=0}^N \overline{a_j} a_k m_{k-j}
    = \int \sum_{j,k=0}^N \overline{a_j} a_k z^{k-j} d\nu
    =\int |p(z)|^2 d\nu.
  \]
  Since there exists $z_0 \in \text{supp}\nu$ such that $p(z_0) \ne 0$, we have $\int |p(z)|^2 d\nu >0$ and hence $H_N(\bm M_N^*(1))$ is positive definite.
\end{proof}

With these preparations at hand, we are ready to prove Theorem \ref{thm:mainexst}.
\begin{proof}[Proof of Theorem \ref{thm:mainexst}]
  If the measure $\nu$ has no more than $N$ points of support, clearly $r=1$ is a solution to Eq.(\ref{eq:poeqrm}) with $\det H_N(\bm M_N^*(1))=0$. The solution is unique because for any $r'<1$, we see from Proposition \ref{prop:eiginc} that $\lambda(r';N)<0$ and hence no measure can generate $\bm M_N^*(r')$, as illustrated in Lemma \ref{thm:toep}.

  Now suppose that $\nu$ has at least $(N+1)$ points of support. If $m_1 = \dots = m_N = 0$, it can be verified that $r=0$, $\sum_{\alpha}\rho_\alpha=m_0$ and arbitrary $\phi_\alpha$ constitute a solution to Eq.(\ref{eq:poeqrm}).
  On the other hand, for any $r'>0$, the equations $\sum_\alpha \rho_\alpha z_\alpha^k = 0$ with $|z_\alpha|=1$ for $k=1,\dots,N$ can only lead to $\rho_\alpha=0$ for all $\alpha$, a contradiction to $m_0>0$.
  Thus, $r=0$ is the only solution in this case.

  If $(m_1,\dots,m_N) \ne \bm 0$, then we claim
  \[
    \lim_{r\to 0^+}\lambda(r;N)<0.
  \]
  To see this, let $k$ be that $m_1=\dots=m_{k-1}=0$ and $m_k\ne 0$. Then the determinant
  \[
    \det(H_k(\bm M_N^*(r)) = \det
    \begin{bmatrix}
      m_0 & 0 & \cdots & 0 & m_k r^{-k} \\
          & m_0 &&& \\
          &   & \ddots && \\
          &   & & m_0 & \\
      \overline{m_k}r^{-k} & 0 & \cdots & 0 & m_0
    \end{bmatrix}
  \]
  has a leading term $-m_0^{k-1}|m_k|^2 r^{-2k}$ and hence goes to $-\infty$ as $r\to 0^+$. This implies that $H_k(\bm M_N^*(r))$ has negative eigenvalues as $r\to 0^+$. Our claim thus follows because the minimum eigenvalue of $H_N(\bm M_N^*(r))$ can never be greater than that of its principal minor.
  Since $\lambda(r;N)$ is strictly increasing on $r\in(0,1]$ (Proposition \ref{prop:eiginc}) and $\lambda(1;N)>0$ (Proposition \ref{prop:pd1}),
  it becomes clear that there exists a unique $r\in (0,1)$, together with a corresponding $n$-atomic measure $\mu = \sum_\alpha \rho_\alpha \delta_{z_\alpha}$ on $\mathbb T$ with $n\le N$, that solves Eq.(\ref{eq:poeqrm}).

  In all three cases, it is straightforward to show that the unclosed terms (for instance, $m_{N+1}$) can be uniquely determined from $m_0,\dots,m_N$. Thus, the Poisson-EQMOM is well-defined.
\end{proof}

\section{Moment inversion algorithm} \label{sec:inv}
Our analysis in the previous section not only reveals the nice property of the Poisson-EQMOM, but also gives clue to a practical algorithm solving Eq.(\ref{eq:poeqrm}). Generally, given any realizable moment $\bm M_N$, the algorithm consists of the following three steps:

(\textit{i}) Find $r\in[0,1]$, which will be detailed in Section \ref{subsec:findr};

(\textit{ii}) Solve $\phi_1,\dots,\phi_N$, as presented in Section \ref{subsec:phik};

(\textit{iii}) Solve $\rho_1,\dots,\rho_N$ from the linear equation
\begin{equation} \label{eq:solverho}
  \begin{bmatrix}
    1 & 1 & \cdots & 1 \\
    z_1 & z_2 & \cdots & z_N \\
    \vdots & \vdots & & \vdots \\
    z_1^{N-1} & z_2^{N-1} & \cdots & z_N^{N-1}
  \end{bmatrix}
  \begin{bmatrix}
    \rho_1 \\ \rho_2 \\ \vdots \\ \rho_N
  \end{bmatrix} =
  \begin{bmatrix}
    m_0 \\ m_1 / r \\ \vdots \\ m_{N-1} / r^{N-1}
  \end{bmatrix}.
\end{equation}
Here we denote $z_\alpha=e^{i\phi_\alpha}$ for $\alpha=1,\dots,N$. In practice, $\rho_\alpha$ can be solved from the real part of the above equation, that is, $\sum_{\alpha=1}^N \rho_\alpha \cos k \phi_\alpha = c_k$ for $k=0,1,\dots,N-1$.

\subsection{Locating $r$} \label{subsec:findr}
Two approaches are available for this purpose. The more effective strategy is to numerically compute the minimum eigenvalue of $H_N(M_N^*(r))$, which is expected to be zero for the desired $r$. As the minimum eigenvalue increases monotonically on $r\in[0,1]$, a simple bisection algorithm on the interval $[0,1]$ can be applied to locate $r$.
In each iteration step, the minimum eigenvalue is computed by the MATLAB function $\tt{eig}$ \cite{MATLAB}.
This is the approach implemented in our simulations (see Table \ref{tab:inversion_time_cost} with $\ell=0$ for the runtime).

On the other hand, $r$ can also be determined with a $t$-polynomial
\[
  p_N(t):=\det H_N(M_N^*(1/t))
\]
of degree $\frac{N^2+2N+\mod(N,2)}{2}$.
Let $t_m$ be the smallest root of $p_N$ on the interval $[1,\infty]$. Then we have $r=t_m^{-1}$, corresponding to the greatest root of $\det H_N(M_N^*(r))$ on $r\in[0,1]$.
To see this, notice that for any smaller root $r'$ (if it ever exists), $H_N(M_N^*(r'))$ must have negative eigenvalues, as indicated by Proposition \ref{prop:eiginc}. Consequently, according to Lemma \ref{thm:toep}, there is no measure with exactly $N$ points of support that generates $M_N^*(r')$.

The polynomials $p_N(t)$ can be explicitly computed for small values of $N$:
\[
\begin{aligned}
  p_1(t) = &m_0^2 - |m_1|^2 t^2, \\
  p_2(t) = &m_0^3 - 2m_0|m_1|^2 t^2 + \left[ 2 \Re (\overline{m_1}^2 m_2) - m_0|m_2|^2 \right] t^4, \\
  p_3(t) = &m_0^4 - 3m_0^2|m_1|^2 t^2 + \left[ 4 \Re (m_0 \overline{m_1}^2 m_2) - 2m_0^2|m_2|^2 + |m_1|^4 \right] t^4 \\
  &+ \left[ 4\Re (m_0\overline{m_1m_2}m_3) - 2 \Re (\overline{m_1}^3m_3) - m_0^2|m_3|^2 - 2|m_1m_2|^2 \right] t^6 \\
  &+ \left[ |m_2|^4 + |m_1m_3|^2 -2\Re (\overline{m_1}m_2^2\overline{m_3}) \right] t^8,
\end{aligned}
\]
with $\Re (z)$ denoting the real part of $z\in\mathbb C$.
Thus, the smallest root $t_m$ on $[1,\infty]$ can be found by existing algorithms.
However, as $N$ grows larger, the explicit forms of $p_N(t)$ will soon become intractable.

\subsection{Solving $\phi_\alpha$'s} \label{subsec:phik}
The key idea to solve the $\phi_\alpha$'s is to seek the roots of orthogonal polynomials.
Given a moment set $\bm M_N \in \mathbb C^{N+1}$, assume that the Toeplitz matrix $H_N(\bm M_N)$ is positive semi-definite with its rank being $n \le N$.
Define $P_0(z)=1$ and
\begin{equation} \label{eq:Pkz}
  P_k(z) = \frac{1}{\det H_{k-1}} \det
  \begin{bmatrix}
    m_0 & m_1 & \cdots & m_{k-1} & m_k \\
    \overline{m_1} & m_0 & \cdots & m_{k-2} & m_{k-1} \\
    \vdots & \vdots & & \vdots & \vdots \\
    \overline{m_{k-1}} & \overline{m_{k-2}} & \cdots & m_0 & m_1 \\
    1 & z & \cdots & z^{k-1} & z^k
  \end{bmatrix}
\end{equation}
for $k=1,\dots,n$. These are monic polynomials of degree $k$.

To see the `orthogonality', we introduce a sesquilinear form $\langle \cdot,\cdot \rangle$ on the polynomial space $\mathbb C[z]_n$ as
\begin{equation} \label{eq:hemform}
  \langle z^j,z^k \rangle = m_{k-j}, \quad j,k=0,1,\dots,n.
\end{equation}
This is actually a complex Hermitian form due to our convention $m_{-k}=\overline{m_k}$.
It is not difficult to verify that
\[
  \langle P_k(z), z^j \rangle = \frac{1}{\det H_{k-1}} \det
  \begin{bmatrix}
    m_0 & \overline{m_1} & \cdots & \overline{m_{k-1}} & \overline{m_k} \\
    m_1 & m_0 & \cdots & \overline{m_{k-2}} & \overline{m_{k-1}} \\
    \vdots & \vdots & & \vdots & \vdots \\
    m_{k-1} & m_{k-2} & \cdots & m_0 & \overline{m_1} \\
    m_j & m_{j-1} & \dots & m_{j-k+1} & m_{j-k}
  \end{bmatrix}
  = \frac{\det H_k}{\det H_{k-1}} \delta_{jk}
\]
for $j=0,1,\dots,k$. Here $\delta_{jk}$ denotes the Kronecker delta. Obviously, this implies
\[
  \langle P_k(z), P_j(z) \rangle = \frac{\det H_k}{\det H_{k-1}} \delta_{jk} \quad \text{and} \quad
  \langle P_n(z), P_n(z) \rangle = 0.
\]

Another ingredient we need is the recursive formula of the orthogonal polynomials.
Define $P_0^*(z)=1$ and
\[
  P_k^*(z) = \frac{1}{\det H_{k-1}}\det
  \begin{bmatrix}
    m_0 & \overline{m_1} & \cdots & \overline{m_{k-1}} & \overline{m_k} \\
    m_1 & m_0 & \cdots & \overline{m_{k-2}} & \overline{m_{k-1}} \\
    \vdots & \vdots & & \vdots & \vdots \\
    m_{k-1} & m_{k-2} & \cdots & m_0 & \overline{m_1} \\
    z^k & z^{k-1} & \cdots & z & 1
  \end{bmatrix}
\]
for $k=1,\dots,n$.
Similarly, one can check that
\[
  \langle P_k^*(z),z^j \rangle = \frac{\det H_k}{\det H_{k-1}}\delta_{0j}
\]
for $j=0,1,\dots,k$.
Then the following relations hold \cite{gtm277}:
\begin{equation} \label{eq:Pkrec}
  P_{k+1}(z) = zP_k(z) - \overline{a_k} P_k^*(z), \quad
  P_{k+1}^*(z) = P_k^*(z) - a_k z P_k(z)
\end{equation}
for $k=0,1,\dots,n-1$ with
\[
  a_k = -\overline{P_{k+1}(0)} = \frac{\det H_{k-1}}{\det H_k} \langle zP_k,1 \rangle.
\]

Now let us specify the arbitrary moment set $\bm M_N$ to be $\bm M_N^*(r)$ in Eq.(\ref{eq:MNstar}), so obviously our assumption on the Toeplitz matrix $H_N(\bm M_N^*(r))$ holds.
If the polynomials $P_k(z)$ in Eq.(\ref{eq:Pkz}) and the Hermitian form in Eq.(\ref{eq:hemform}) are both defined based on $\bm M_N^*(r)$, then we claim the following.
\begin{proposition}
  The roots of $P_n(z)$ are just $z_{\alpha}=e^{i\phi_\alpha}$ for $\alpha=1,\dots,n$.
\end{proposition}

\begin{proof}
  Denote $m_k^*=m_k r^{-k}$.
  According to Lemma \ref{thm:toep}, there exists a unique $n$-atomic measure $\mu = \sum_{\alpha=1}^n \rho_\alpha \delta_{z_\alpha}$ on $\mathbb T$ such that $m_k^* = \int z^k d\mu$, where $z_\alpha=e^{i\phi_\alpha}$ are just what we need.
  Then the Hermitian form Eq.(\ref{eq:hemform}) is rewritten as $\langle z^j, z^k \rangle = m^*_{k-j} = \int z^{k-j} d\mu$.
  It is straightforward to verify that for any $p,q\in \mathbb C[z]_n$, we have
  \[
    \langle p,q\rangle = \int \bar p q d\mu.
  \]
  In particular, it is seen that
  \[
    \langle P_n(z),P_n(z) \rangle = \int |P_n(z)|^2 d\mu = 0,
  \]
  which implies $P(z_\alpha)=0$ for all $\alpha$.
\end{proof}

For convenience, define $Q_0(z) = 1$ and $Q_k(z) = \sqrt{\det H_{k-1}/\det H_k} P_k(z)$ for $k=1,\dots,n-1$. Clearly, we have $\langle Q_k(z), Q_j(z) \rangle = \delta_{kj}$, indicating that $\{Q_k\}_{k=0}^{n-1}$ is the (unique) orthonormal base of $\mathbb C[z]_{n-1}$.
As a consequence, we can express $zQ_k(z)$ as
\[
  zQ_k(z) = \sum_{j=0}^{k+1} a_{kj} Q_j(z)
\]
for $k=0,1,\dots,n-2$. A similar relation for $zQ_{n-1}(z)$ can be obtained from Eq.(\ref{eq:Pkrec}) as
\[
  zQ_{n-1}(z) = \sum_{j=0}^{n-1}a_{n-1,j}Q_j(z) + \sqrt{\frac{\det H_{n-2}}{\det H_{n-1}}} P_n(z),
\]
because $P_{n-1}^*(z)$ can also be linearly generated by $\{Q_k\}_{k=0}^{n-1}$.
The last two relations can be concisely organized as
\[
  z
  \begin{bmatrix}
    Q_0(z) \\ \vdots \\ Q_{n-2}(z) \\ Q_{n-1}(z)
  \end{bmatrix}
  = \bm J
  \begin{bmatrix}
    Q_0(z) \\ \vdots \\ Q_{n-2}(z) \\ Q_{n-1}(z)
  \end{bmatrix}
  +
  \begin{bmatrix}
    0 \\ \vdots \\ 0 \\ \sqrt{\frac{\det H_{n-2}}{\det H_{n-1}}} P_n(z)
  \end{bmatrix}
\]
with $\bm J=(a_{kj}) \in \mathbb C^{n \times n}$ and $a_{kj}=0$ for $j\ge k+2$.
It thus becomes clear that the roots of $P_n(z)$ are the eigenvalues of $\bm J$.

The elements of $\bm J$ can be computed iteratively. If $Q_k(z)$ is already known, then we have
\begin{equation} \label{eq:akj}
  a_{kj} = \overline{\langle zQ_k(z),Q_j(z)\rangle} \quad \text{for } j=0,\dots,k.
\end{equation}
Compute
\begin{equation}
  R_{k+1}(z) = zQ_k(z) - \sum_{j=0}^k a_{kj} Q_j(z),
\end{equation}
which has a degree $k+1$ and is orthogonal to $Q_j(z)$ for $j=0,\dots,k$.
Therefore, $Q_{k+1}(z)$ can be obtained via
\begin{align}
  a_{k,k+1} &= \sqrt{\langle R_{k+1}(z),R_{k+1}(z) \rangle}, \\
  Q_{k+1}(z) &= \frac{R_{k+1}(z)}{a_{k,k+1}}. \label{eq:Qkp1}
\end{align}
This gives a practical algorithm for solving the $\phi_\alpha$'s with a given moment set $M_N^*(r)$.

\begin{remark}
  For the degenerated case with $n = \text{rank } M_N^*(r)<N$, only $n$ nodes $z_{\alpha}=e^{i\phi_\alpha}$ ($\alpha=1,\dots,n$) can be solved from the above procedure, and the rest nodes can be chosen arbitrarily with zero weights.
\end{remark}

\subsection{Lifting $m_0$} \label{subsec:lift}
Two issues remain for the moment inversion process presented in the previous subsections. First, it may be computationally costly to solve $r$ iteratively as in subsection \ref{subsec:findr}. Second, the approximation $f_N$ may deviate significantly from the `correct' distribution $f$ especially when the number of nodes $N$ is not too large. The latter is clearly illustrated in Fig.~\ref{fig:lift}, where the original von Mises distribution reads as
\[
  f_{vM}(\theta) = C\exp\left( \frac{\cos\theta}{d} \right)
\]
with $C$ being a scaling factor such that $\int_{-\pi}^\pi f_{vM}d\theta = 1$.
As seen in Fig.~\ref{fig:lift}(A), if $d=0.8$ (corresponding to a `wide' distribution), the 8-node Poisson-EQMOM (with $\ell=0$) generates a reasonable approximation to $f_{vM}$.
In contrast, for the `narrower' distribution with $d=0.4$, the 8-node reconstruction is very different from the von Mises distribution; see Fig.~\ref{fig:lift}(B) and the large $L^2$-error in Fig.~\ref{fig:lift}(D) for $\ell=0$.
While simply increasing the number of nodes does imply a trend for $f_N$ to be closer to $f_{vM}$, there remains a non-negligible error even when $N$ reaches 32.
This poses challenges to solving the moment closure system (\ref{eq:momsys}) as the closed moment $\hat m_{N+1}$ may become inaccurate, as demonstrated in Table~\ref{tab:err} for the von Mises case with $d=0.4$.

\begin{figure}[htbp]
  \centering
  \subfloat[$d=0.8$ and $N=8$]
  {\includegraphics[height=4.5cm]{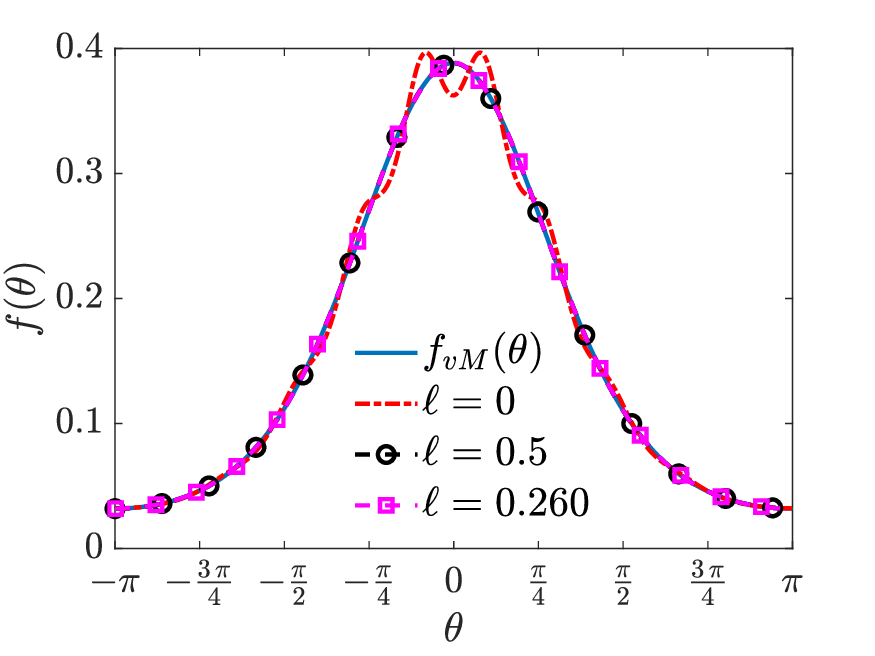}}
  \subfloat[$d=0.4$ and $\ell=0$]
  {\includegraphics[height=4.5cm]{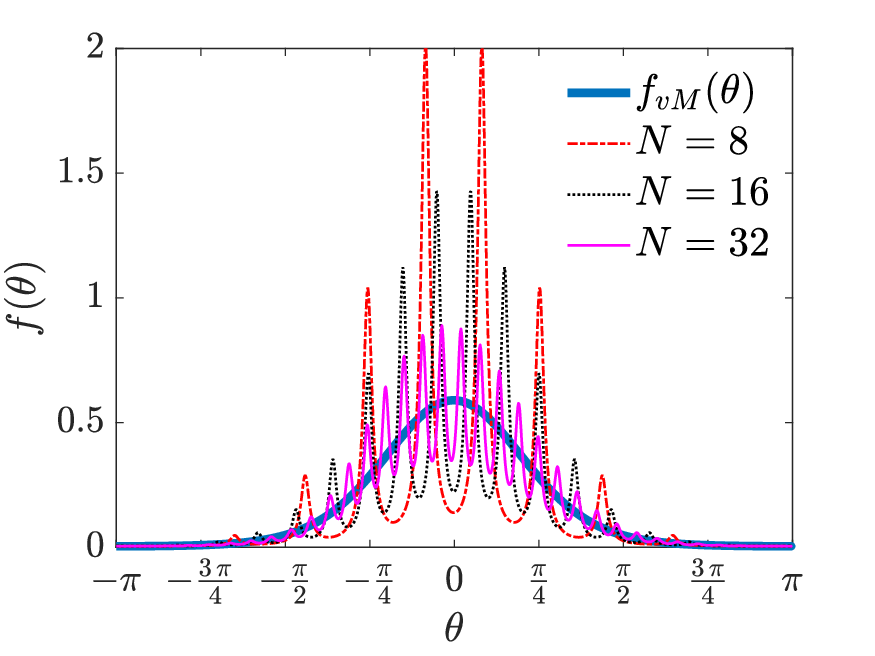}}
  \vspace{0.5pt}
  \subfloat[$d=0.4$ and $N=8$]
  {\includegraphics[height=4.5cm]{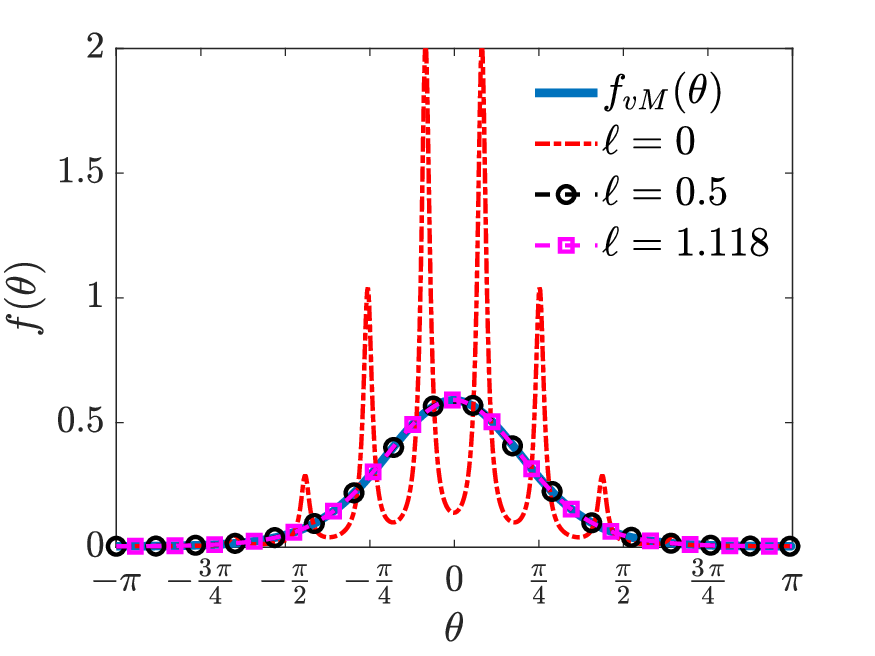}}
  \subfloat[$L^2$-error for the cases with $d=0.4$]
  {\includegraphics[height=4.5cm]{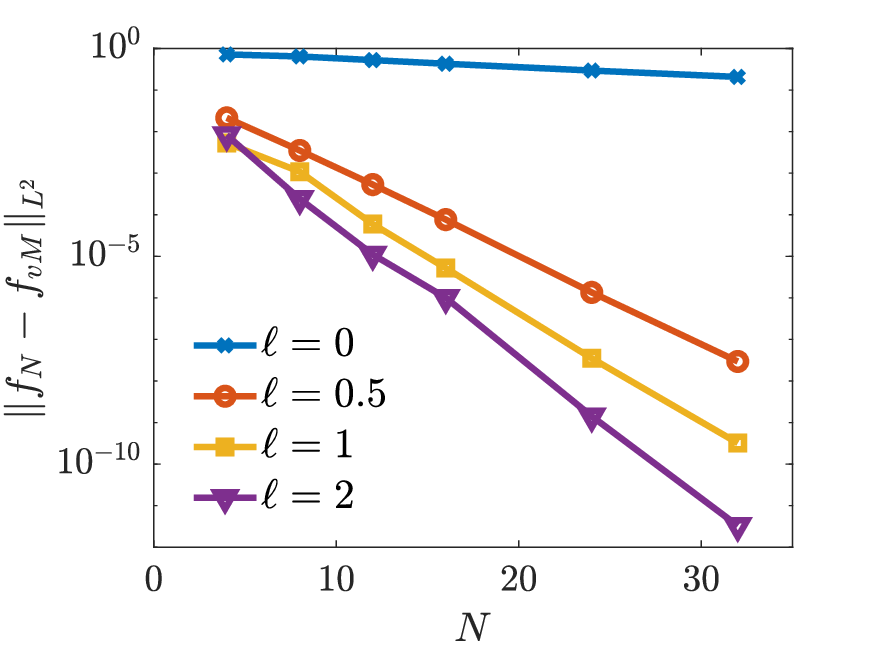}}
  \caption{Reconstruction of the von Mises distribution by the Poisson-EQMOM.}
  \label{fig:lift}
\end{figure}

\begin{table}[htbp]
  \centering
  \caption{$|\hat m_{N+1} - m_{N+1}|$ for the von Mises distribution with $d=0.4$}
  \label{tab:err}
  \begin{tabular}{ccccc}
    \hline\noalign{\smallskip}
    $N$ & $\ell=0$ & $\ell=0.5$ & $\ell=1$ & $\ell=2$ \\
    \noalign{\smallskip}\hline\noalign{\smallskip}
    4 & 0.1154 & 0.01698 & 0.005697 &	0.01389 \\
    8 & 0.09176 &	0.001776 & 0.001535 &	$3.694\times 10^{-4}$ \\
    12 & 0.07687 & $1.599\times 10^{-4}$ & $1.011\times 10^{-6}$ & $3.685\times 10^{-6}$ \\
    16 & 0.06484 & $8.005\times 10^{-6}$ & $8.880\times 10^{-6}$ & $1.727\times 10^{-6}$ \\
    24 & 0.04627 & $6.069\times 10^{-7}$ & $4.172\times 10^{-8}$ & $7.589\times 10^{-10}$ \\
    32 & 0.04420 & $3.727\times 10^{-8}$ & $1.018\times 10^{-10}$	& $4.056 \times 10^{-12}$ \\
    \noalign{\smallskip}\hline
  \end{tabular}
\end{table}

Inspired by the condition in Theorem \ref{thm:lift} (that $f$ has a positive lower bound), we propose a remedy, which is simply raising $m_0$ by some quantity $\ell$ (with other moments unchanged) before performing moment inversion. This is equivalent to inverting the moments generated by $f(\theta)+\ell$ because the constant function on $\mathbb T$ has zero moments $m_k=0$ for $k\ge 1$.
With the lifted moments $(m_0+\ell,m_1,\dots,m_N)$, the inverted parameters are denoted as $(\rho_\alpha^{(\ell)},\phi_\alpha^{(\ell)},r^{(\ell)})_{\alpha=1}^N$, all relying on $\ell$.
The Poisson-EQMOM ansatz Eq.(\ref{eq:peans}) now becomes
\begin{equation} \label{eq:peanslift}
  f_{N,\ell}(\theta) = \sum_{\alpha=1}^N \rho_\alpha^{(\ell)} P_{r^{(\ell)}} (\phi_\alpha^{(\ell)}-\theta) - \ell
\end{equation}
for moment closure. In particular, $\hat m_{N+1}$ has the same form as Eq.(\ref{eq:mnp1}).

It is seen from Figs.~\ref{fig:lift}(A\&C) that with certain values of lifting $\ell$, the recovered $f_{N,\ell}$'s are much more precise especially for the case with $d=0.4$. The improvement is more clearly manifested in Fig.~\ref{fig:lift}(D) and Table~\ref{tab:err}: with the presence of lifting, the $L^2$-error is reduced by several orders for a fixed $N$, and can be more effectively eliminated by adding $N$ as compared with the zero lifting cases.
Such an observation is consistent with the pointwise convergence of $f_N$ to $f$ in Theorem \ref{thm:lift}.
Therefore, we believe that lifting $m_0$ is an indispensable part of the moment inversion process.

The mechanism of lifting $m_0$ can be understood as follows. According to Proposition \ref{prop:eiginc}, the minimum eigenvalue $\lambda(r;N)$ of the Toeplitz matrix $H_N(\bm M_N^*(r))$ (see (\ref{eq:2.13}) \& (\ref{eq:MNstar})) increases on $r$. In some cases, $r$ may have to be so close to $1$ to ensure $\lambda(r;N)=0$ and cause the resultant Poisson kernel akin to a Delta function. This could interpret the problem observed in Fig.~\ref{fig:lift}(B). Through lifting $m_0$ to $m_0+\ell$ (while keeping the other moments unchanged), it is easy to show that $r$ solves $\lambda(r;N)=-\ell$. Thus, $r$ becomes smaller and the accuracy of function approximation is promoted, as seen in Fig.~\ref{fig:lift}(C).
As such, the lifting can be seen as a numerical regularization technique that overcomes some ill-conditioning of the moment inversion. This may be compared to approaches to regularize/approximate standard $M_N$ methods \cite{All2019,McD2013,SchTor2015}. It also shares common features with the filter method in the literature \cite{fan2020} as both methods modify the moments and the computation of the moments becomes more stable in the modified setting.

Based on this reasoning, it is possible to choose $\ell$ flexibly to accelerate the calculation of $r$. Given $\bm M_N = (m_0,\dots,m_N)\in\mathbb C^{N+1}$, we set $r_1=|m_1|/m_0$ and the eigenvalues of $H_1(\bm M^*_N(r_1)) = \begin{bmatrix} m_0 & m_1/r_1 \\ \overline{m_1}/r_1 & m_0 \end{bmatrix}$ are 0 and $2m_0$. Then for $H_N(\bm M^*_N(r_1))$, we have the minimum eigenvalue $\lambda(r_1;N) \le 0$ and take
\begin{equation} \label{eq:liftup}
    \ell= \min \{ -\lambda(r_1;N), \ \ell_M \}.
\end{equation}
Here $\ell_M$ ($=10^6$ in this work) is a prescribed limit value.
It is easy to see that for $\bm M_{N\ell}=(m_0+\ell,m_1,\dots,m_N)$, the minimum eigenvalue of $H_N(\bm M_{N\ell}^*(r_1))$ is zero, so the remaining task is to solve $\phi_\alpha$ and $\rho_\alpha$ based on $\bm M^*_{N\ell}(r_1)$.
Clearly, this approach computes the minimum eigenvalue of a Toeplitz matrix only once.
The performance is examined in Figs.~\ref{fig:lift}(A) (the case $\ell=0.260$) and (C) (the case $\ell=1.118$).

Table~\ref{tab:inversion_time_cost} quantifies the runtime of performing moment inversion by 10 000 times with no lifting ($\ell=0$) and with the lifting based on (\ref{eq:liftup}). Here the moments $m_k$ ($k=0,\dots,N$) are generated by the parametric distributions
\[
    f_r(\theta) = \frac{1}{4\pi}\left[ 1 + \sum_{j = 1}^6 P_j \sin(p_j \theta) \right] + \frac{1}{4\pi} \frac{1 - R^2}{1 - 2R\cos\theta + R^2}
\]
with $\{p_j\}_{1}^6 = \{1,2,3,5,7,11\}$, $P_j$'s uniformly sampled from $[-0.5,0.5]$ and $R$ uniformly sampled from $[0,1]$. Note that for $\ell=0$, the first approach in Section \ref{subsec:findr} is applied to locate $r$ through iteration. The superlinear growth of the runtime with respect to $N$ can be observed for both values of $\ell$. As expected, the lifting (\ref{eq:liftup}) substantially reduces the runtime. It is worth pointing out that the moment inversion without lifting could fail for $N=32$, which is caused by the numerical ill-conditioning of computing the roots of orthogonal polynomials.
In contrast, the procedure with lifting is more robust against the increment of $N$, which may be attributed to the fact that the higher moments on $\mathbb S^1$ (Fourier coefficients) are small.
While the stability of the lifted method is improved and we did not observe any stability problems up to $N = 32$, the stability of the lifted method is not guaranteed for arbitrary $N$.

Finally, the established moment inversion algorithm is summarized in Table~\ref{tab:sum}. In particular, the choice of $\ell$ by (\ref{eq:liftup}) is applied in our numerical tests in Section \ref{sec:num} and yields satisfactory results therein.

\begin{table}[htbp]
    \caption{The runtime (in seconds) of performing moment inversion by 10 000 times without lifting ($\ell=0$) and with lifting (\ref{eq:liftup}). The results are obtained by the MATLAB code running on an Intel i7-1065G7 CPU.}
    \label{tab:inversion_time_cost}
    \centering
    \begin{threeparttable}
    \begin{tabular}{c|cccccccc}
    \hline
    $N$ & 4 & 8 & 12 & 16 & 20 & 24 & 28 & 32 \\
    \hline
    $\ell=0$ & 14.4332 & 17.1844 & 20.5417 & 24.6429 & 154.5748 & 189.8994 & 226.0661 & --\tnote{1} \\
    $\ell$ by (\ref{eq:liftup}) & 1.5987 & 2.9869 & 4.6814 & 7.4622 & 15.6160 & 19.9326 & 26.1912 & 36.0146 \\
    \hline
    \end{tabular}

    \begin{tablenotes}
        \footnotesize
        \item[1] Computation could crash for some of the cases.
    \end{tablenotes}
    \end{threeparttable}
\end{table}

\begin{table}[htbp]
  \centering
  \caption{A summary on the moment inversion algorithm}
  \label{tab:sum}
  \begin{tabular}{ll}
    \hline\noalign{\smallskip}
    Step 0 & Given $\bm M_N=(m_0,\dots,m_N)\in\mathbb C^{N+1}$. \\
    Step 1 & Compute $r = |m_1|/m_0$ and $\ell$ by Eq.(\ref{eq:liftup}). \\
    Step 2 & $\bm M_N \gets (m_0+\ell, \frac{m_1}{r},\dots, \frac{m_N}{r^N})$.\\
    Step 3 & Define the Hermitian form $\langle z^j,z^k\rangle = m_{k-j}$. \\
    Step 4 & Compute the elements of $\bm J$ iteratively by Eqs.(\ref{eq:akj})-(\ref{eq:Qkp1}). \\
    Step 5 & Solve $z_\alpha=e^{i\phi_\alpha}$ ($\alpha=1,\dots,N$) as the eigenvalues of $\bm J$. \\
    Step 6 & Solve $\rho_\alpha$ ($\alpha=1,\dots,N$) by Eq.(\ref{eq:solverho}). \\
    \noalign{\smallskip}\hline
  \end{tabular}
\end{table}

\section{Application to polar active flows} \label{sec:app}

We now apply the Poisson-EQMOM to a 2D kinetic equation associated with the celebrated Vicsek model \cite{Vicsek95}. The evolution of $f=f(t,\bm x,\theta)$ is governed by \cite{degond2008}
\begin{equation} \label{eq:vic}
  \epsilon (\partial_t f + \bm e_{\theta} \cdot \nabla_{\bm x} f) = Q(f):= \mathtt d \partial_\theta^2 f + \nu \partial_\theta \left((\Omega \times \bm e_\theta)f \right).
\end{equation}
In the Vicsek dynamics, each polar active particle continuously adjusts its velocity direction to the neighboring mean velocity direction
\begin{equation} \label{eq:Omega}
  \Omega = \Omega[f](t,\bm x) = \frac{j[f](t,\bm x)}{|j[f](t,\bm x)|}\in S^1 \quad
  \text{with }
  j[f](t,\bm x) = \int_{-\pi}^{\pi} \bm e_{\theta} f(\theta) d\theta.
\end{equation}
Writing $\Omega=(\Omega_x,\Omega_y)=(\cos\bar\theta,\sin\bar\theta)$, we see that the cross product
\[
  \Omega \times \bm e_\theta = \Omega_x \sin \theta - \Omega_y \cos \theta = \sin(\theta-\bar\theta)
\]
quantifies the angle between the single particle velocity $\bm e_\theta$ and the neighboring mean direction $\Omega$. $\nu>0$ is the constant intensity of velocity alignment.
But the alignment is not perfect. The competing mechanism is the Gaussian (or uniform) noise on $[-\pi,\pi]$, with $\mathtt d>0$ characterizing the noise strength.
The whole system can be rescaled and produces a scaling parameter $\epsilon$ as the ratio between micro and macro variables. The hydrodynamic limit is derived with $\epsilon \to 0$.

As a remarkable feature of Eq.(\ref{eq:vic}), the collision operator $Q(f)$ can be rewritten as \cite{degond2008}
\[
  Q(f) = \frac{\mathtt d}{\nu} \partial_\theta \left( M_{\bar\theta} \partial_\theta \left( \frac{f}{M_{\bar\theta}} \right) \right)
\]
with $M_{\bar\theta}=M_{\bar\theta}(\theta)$ being the von Mises distribution:
\begin{equation} \label{eq:vmeq}
  M_{\bar\theta}(\theta) = C\exp \left( \frac{\nu \cos (\theta-\bar\theta)}{\mathtt d} \right).
\end{equation}
It holds that
\[
  H(f):=\int_{-\pi}^\pi Q(f) \frac{f}{M_{\bar\theta}} d\theta = -\int_{-\pi}^\pi M_{\bar\theta} \partial_\theta^2 \left( \frac{f}{M_{\bar\theta}} \right) d\theta \le 0.
\]
Therefore, the equilibrium that vanishes $Q(f)$ can only be von Mises distributions.

The moment system derived from the $N$-node Poisson-EQMOM has the form of Eq.(\ref{eq:momsys}) (with $v_0=1$ and $\epsilon$ included).
The source term  $S(\bm M_N)=(q_0,q_1,\dots,q_N)$ can now be determined by a direct calculation:
\[
  q_k:=\int_{-\pi}^\pi e^{ik\theta} Q(f) d\theta = \mathtt d \int_{-\pi}^{\pi} e^{ik\theta} \partial_\theta^2 f d\theta
  + \nu \int_{-\pi}^{\pi} e^{ik\theta} \partial_\theta \left((\Omega \times \bm e_\theta)f \right) d\theta.
\]
Using the integration by parts gives
\[
  \int_{-\pi}^{\pi} e^{ik\theta} \partial_\theta^2 f d\theta = (ik)^2 \int_{-\pi}^{\pi} f e^{ik\theta} d\theta = -k^2 m_k
\]
and
\[
\begin{aligned}
  \int_{-\pi}^{\pi} e^{ik\theta} \partial_\theta \left((\Omega \times \bm e_\theta)f \right) d\theta &= -ik \int_{-\pi}^{\pi} e^{ik\theta} (\Omega_x \sin \theta - \Omega_y \cos \theta) f d\theta \\
  &= -ik \left[ \frac{\Omega_x}{2i}(m_{k+1}-m_{k-1}) - \frac{\Omega_y}{2}(m_{k+1}+m_{k-1}) \right] \\
  &= \frac{k}{2} \left[ (\Omega_x+i\Omega_y)m_{k-1} - (\Omega_x-i\Omega_y)m_{k+1} \right] \\
  &= \frac{k}{2|m_1|} (m_1 m_{k-1} - \overline{m_1} m_{k+1}).
\end{aligned}
\]
Here we have used the identity $\Omega_x+i\Omega_y=m_1/|m_1|$, which can be easily seen from Eq.(\ref{eq:Omega}).
Therefore, the source term $S(\bm M_N)$ is obtained with
\begin{equation}
  q_k = -\mathtt d k^2 m_k + \mathbf 1_{|m_1|>0} \cdot \frac{\nu k}{2|m_1|}(m_1 m_{k-1} - \overline{m_1} m_{k+1})
\end{equation}
for $k=0,\dots,N$. Clearly, no other term except for $m_{N+1}$ needs to be closed, and $q_0=0$ indicates that the system respects mass conservation.

\subsection{Hyperbolicity with $N=1$}
We briefly discuss the hyperbolicity of the resultant moment system as first-order PDEs.
For the single-node system, the unknown is $\bm c = (c_0,c_1,s_1)^T$, and the closure is performed such that $\rho_1=m_0$, $r=|m_1|/m_0$, $e^{i\phi_1}=m_1/|m_1|$ and hence
\[
  m_2 = r^2\rho_1 e^{2i\phi_1} = \frac{m_1^2}{m_0} \quad \Leftrightarrow \quad
  c_2 = \frac{c_1^2-s_1^2}{c_0}, \
  s_2 = \frac{2c_1s_1}{c_0}.
\]
The resultant moment system is explicitly presented as
\begin{equation} \label{eq:momN1}
  \partial \bm c + A_x \partial_x \bm c + A_y \partial_y \bm c = \frac{1}{\epsilon}(0,\Re(q_1), \Im(q_1))^T
\end{equation}
with
\[
  A_x =
  \begin{bmatrix}
    0 & 1 & 0 \\
    \frac{1}{2}-\frac{c_1^2-s_1^2}{2c_0^2} & \frac{c_1}{c_0} & -\frac{s_1}{c_0} \\
    -\frac{c_1s_1}{c_0^2} & \frac{s_1}{c_0} & \frac{c_1}{c_0}
  \end{bmatrix},
  \quad A_y =
  \begin{bmatrix}
    0 & 0 & 1\\
    -\frac{c_1s_1}{c_0^2} & \frac{s_1}{c_0} & \frac{c_1}{c_0} \\
    \frac{1}{2} + \frac{c_1^2-s_1^2}{2c_0^2} & -\frac{c_1}{c_0} & \frac{s_1}{c_0}
  \end{bmatrix}
\]
and
\[
\begin{aligned}
  \Re(q_1) &= -\mathtt d c_1 + \mathbf 1_{|m_1|>0} \cdot \frac{\nu}{2|m_1|} (c_0c_1 - c_1c_2 - s_1s_2), \\
  \Im(q_1) &= -\mathtt d s_1 + \mathbf 1_{|m_1|>0} \cdot \frac{\nu}{2|m_1|} (c_0s_1 - c_1s_2 + s_1c_2)
\end{aligned}
\]
being the real and imaginary parts of $q_1$, respectively.
Our result is stated as follows.
\begin{proposition}
  The moment system (\ref{eq:momN1}) is strictly hyperbolic if $0 \le r < 1$.
\end{proposition}

\begin{proof}
  For any $a,b\in \mathbb R$ with $a^2+b^2>0$, the characteristic polynomial $p(\lambda)$ of $aA_x+bA_y$ can be derived through a straightforward calculation as
  \[
    p(\lambda) = \left( \lambda - \frac{ac_1+bs_1}{c_0} \right) \left( \lambda^2 - \frac{ac_1+bs_1}{c_0}\lambda - \frac{a^2+b^2}{2}(1-r^2) \right),
  \]
  where we have used the relation $r^2=(c_1^2+s_1^2)/c_0^2$.
  Since $(a^2+b^2)(1-r^2)>0$, it is not difficult to verify that $p(\lambda)$ has three distinct real roots.
\end{proof}

  If $r=1$, $f$ has a single support $f(\theta)=c_0\delta(\theta-\theta_0)$ (see the proof of Theorem \ref{thm:mainexst}), implying $c_1=c_0\cos\theta_0$ and $s_1=c_0\sin\theta_0$.
  In this case, the system (\ref{eq:momN1}) is not hyperbolic.

  To see this, just take $a=1$, $b=0$ and notice from the expression of $p(\lambda)$ that $c_1/c_0=\cos\theta_0$ is an eigenvalue of $A_x$ with the algebraic multiplicity 2.
  Then the corresponding eigenvector $u=(u_1,u_2,u_3) \in \mathbb R^3$ can be solved from $A_x u = u\cos\theta_0$, yielding $u_2=u_1\cos\theta_0$ and $u_3=u_1\sin\theta_0$ (if $\sin\theta_0\ne 0$).
  This shows that the eigenvalue $\cos\theta_0$ has the geometric multiplicity 1, and hence $A_x$ is not diagonalizable if $\sin\theta_0\ne 0$.
  But the case $\sin\theta_0=0$ (that is, $s_1=0$ and $c_1=c_0$) simply renders $A_y$ not diagonalizable.

\begin{remark}
  The non-hyperbolicity with the use of delta functions for moment closure, a method sometimes termed the quadrature method of moments (QMOM), has also been observed in the treatment of Boltzmann-like equations \cite{Chalons2012,Huang2020}.
  The hyperbolicity for $N>1$ remains unclear for the Poisson-EQMOM system, though it is strongly believed to be true since our numerical results with multiple nodes never exhibit any non-hyperbolic shocks (see Section \ref{sec:num}). An analytic proof is left for future work.
\end{remark}

\section{Numerical scheme} \label{sec:sch}
The numerical scheme for the moment closure system involves a splitting method which solves separately
\begin{align}
  \partial_t m_k + \frac{v_0}{2}\partial_x(m_{k+1}+m_{k-1})
  + \frac{v_0}{2i}\partial_y (m_{k+1}-m_{k-1}) &= 0, \label{eq:trans} \\
  \partial_t m_k &= \frac{1}{\epsilon} q_k(\bm M_N) \label{eq:collision}
\end{align}
for $k=0,\dots,N$. The splitting is helpful to handle the small relaxation times (see, e.g., \cite{GAMBA2015}). While the scheme proposed in this section is believed to be asymptotically stable, a rigorous analysis is beyond the scope of this paper and is thus left for future investigations.

\subsection{Transport and kinetic-based flux}
Eq.(\ref{eq:trans}) is treated with the first-order schemes (to simplify the notation, only the spatial 1D case is illustrated here):
\begin{equation}
  \frac{m_{k,j}^*-m_{k,j}^n}{\Delta t} + \frac{v_0}{2} \frac{\mathcal F^n_{k+1,j+\frac{1}{2}} - \mathcal F^n_{k+1,j-\frac{1}{2}}}{\Delta x}
  + \frac{v_0}{2} \frac{\mathcal F^n_{k-1,j+\frac{1}{2}} - \mathcal F^n_{k-1,j-\frac{1}{2}}}{\Delta x} = 0,
\end{equation}
in which the numerical `kinetic-based' fluxes are evaluated based on the approximation $f_N$ \cite{deshpande1986,MarFox2013}. For instance, we have
\[
  \mathcal F_{k,j+\frac{1}{2}} = \int_{(-\frac{\pi}{2},\frac{\pi}{2})} e^{ik\theta} f_j(\theta) d\theta
  + \int_{(-\pi,-\frac{\pi}{2}) \cup (\frac{\pi}{2},\pi)} e^{ik\theta} f_{j+1}(\theta) d\theta.
\]
Substituting the ansatz Eq.(\ref{eq:peans}) or (\ref{eq:peanslift}) for $f(\theta)$ to the above integrands, we see that only the integrals of the form $\int e^{ik\theta} P_r(\phi-\theta)d\theta$ and $\ell \int e^{ik\theta}d\theta$ need to be computed.
For the latter type, it is easy to derive
\[
  \ell \int_{-\frac{\pi}{2}}^{\frac{\pi}{2}} e^{ik\theta}d\theta
  = \frac{\ell}{ik} \left( e^{ik\frac{\pi}{2}} - e^{-ik\frac{\pi}{2}} \right)
  = \frac{2\ell}{k} \sin \left( \frac{k\pi}{2} \right).
\]

In what follows we deduce an analytical expression for
\[
  I_k:=\int_{-\frac{\pi}{2}}^{\frac{\pi}{2}} e^{ik\theta} P_r(\phi-\theta)d\theta,
\]
whereas the integration on other types of intervals can be handled similarly.

With the Fourier series of the Poisson kernel $P_r(\phi-\theta) = \frac{1}{2\pi} \sum\limits_{m\in\mathbb Z} r^{|m|}e^{-im(\phi-\theta)}$, we obtain
\[
\begin{aligned}
  I_k &= \frac{1}{2\pi}\sum_{m\in\mathbb Z} r^{|m|}e^{-im\phi} \int_{-\frac{\pi}{2}}^{\frac{\pi}{2}} e^{i(k+m)\theta} d\theta \\
  &= \frac{r^{|k|}e^{ik\phi}}{2} + \frac{1}{2\pi}\sum_{m\ne-k} \frac{r^{|m|}e^{-im\phi}\left( e^{i(k+m)\frac{\pi}{2}} - e^{-i(k+m)\frac{\pi}{2}} \right)}{i(k+m)} \\
  &= \frac{r^{|k|}e^{ik\phi}}{2} + \frac{e^{ik\phi}}{2\pi i}\sum_{m\ne 0}\frac{r^{|m-k|}e^{-im\phi} \left( e^{im\frac{\pi}{2}} - e^{-im\frac{\pi}{2}} \right) }{m}.
\end{aligned}
\]
With $k>0$, it is clear that we need to compute summations of the form
\[
\begin{aligned}
  \sum_{m\ne 0} \frac{r^{|m-k|}z^m}{m}
  &= \sum_{m<0} \frac{r^{k-m}z^m}{m} + \sum_{m=1}^k \frac{r^{k-m}z^m}{m} + \sum_{m\ge k+1} \frac{r^{m-k}z^m}{m} \\
  &= \sum_{m<0} \frac{r^{k-m}z^m}{m} + \sum_{m=1}^k \frac{\left( r^{k-m}-r^{m-k} \right) z^m}{m} + \sum_{m\ge 1} \frac{r^{m-k}z^m}{m} \\
  &= r^k \ln(1-rz^{-1}) + \sum_{m=1}^k \frac{\left( r^{k-m}-r^{m-k} \right) z^m}{m} - r^{-k} \ln(1-rz).
\end{aligned}
\]
The last equality follows from the formal series for $|z|<1$:
\[
  G(z) = \sum_{m\ge 1} \frac{z^m}{m} = -\ln(1-z).
\]
To see this, just notice $G'(z) = \sum_{m\ge 0} z^m = (1-z)^{-1}$.

As a consequence, we further proceed $I_k$ as
\[
\begin{aligned}
  &\sum_{m\ne 0}\frac{r^{|m-k|}e^{-im\phi} \left( e^{im\frac{\pi}{2}} - e^{-im\frac{\pi}{2}} \right) }{m} \\
  = & -r^k \ln \frac{1-ire^{i\phi}}{1+ire^{i\phi}} + r^{-k} \ln \frac{1+ire^{-i\phi}}{1-ire^{-i\phi}} + \sum_{m=1}^k \left( i^m - (-i)^m \right)  \frac{r^{k-m}-r^{m-k}}{m} e^{-im\phi}.
\end{aligned}
\]

If $k<0$ (note that $m_{-1}$ is involved in the system), $I_k$ can be computed by the same technique, with the results presented below.
\[
\begin{aligned}
  &\sum_{m\ne 0}\frac{r^{|m-k|}e^{-im\phi} \left( e^{im\frac{\pi}{2}} - e^{-im\frac{\pi}{2}} \right) }{m} \\
  = & -r^k \ln \frac{1-ire^{i\phi}}{1+ire^{i\phi}} + r^{-k} \ln \frac{1+ire^{-i\phi}}{1-ire^{-i\phi}} + \sum_{m=1}^{-k} \left( i^m - (-i)^m \right)  \frac{r^{-k-m}-r^{m+k}}{m} e^{im\phi}.
\end{aligned}
\]

\subsection{Collision term} \label{subsec:col}
Eq.(\ref{eq:collision}) is treated with a first-order semi-implicit scheme, which reads as
\[
  \frac{\epsilon}{\Delta t} \left( m_k^{n+1} - m_k^* \right) = -\mathtt d k^2 m_k^{n+1} + \mathbf 1_{|m_1^*|>0} \cdot \frac{\nu k}{2|m_1^*|} \left(m_1^* m_{k-1}^{n+1} - \overline{m_1}^* m_{k+1}^{n+1} \right)
\]
for $k=0,1,\dots,N-1$, and
\[
  \frac{\epsilon}{\Delta t} \left( m_N^{n+1} - m_N^* \right) = -\mathtt d N^2 m_N^{n+1} + \mathbf 1_{|m_1^*|>0} \cdot \frac{\nu N}{2|m_1^*|} \left(m_1^* m_{N-1}^{n+1} - \overline{m_1}^* \hat m_{N+1}^n \right)
\]
for the last moment $m_N$.
The above equations are concisely collected as
\[
  \begin{bmatrix}
    \gamma & 0 &&& \\
    \ddots & \ddots & \ddots && \\
    & -\alpha k & \gamma + dk^2 & \bar\alpha k & \\
    && \ddots & \ddots & \ddots \\
    &&& -\alpha N & \gamma + dN^2
  \end{bmatrix}
  \begin{bmatrix}
    m_0^{n+1} \\ \vdots \\ m_k^{n+1} \\ \vdots \\ m_N^{n+1}
  \end{bmatrix}
  = \gamma
  \begin{bmatrix}
    m_0^* \\ \vdots \\ m_k^* \\ \vdots \\ m_N^*
  \end{bmatrix}
  +
  \begin{bmatrix}
    0 \\ \vdots \\ 0 \\ \vdots \\ -\bar\alpha N \hat m_{N+1}^n
  \end{bmatrix}
\]
with $\alpha=\frac{\nu m_1^*}{2|m_1^*|}$ and $\gamma=\frac{\epsilon}{\Delta t}$.

During the simulation, the time step $\Delta t$ is chosen to satisfy the CFL condition while ensuring that the above coefficient matrix is strictly diagonally dominant.

\subsection{Boundary conditions}
Three kinds of boundary conditions are used in this work, including the Neumann condition, the periodic boundary and the reflexive wall boundary.
For the last one, it is assumed that all particles rebound elastically with no flux into the wall. Let $x_b$ be the 1-D wall position and $x>x_b$ the physical domain.
Then at the kinetic level, the reflexive condition reads as
\[
  f(x_b,\theta) = f(x_b,\pi-\theta) \quad \text{for} \quad -\frac{\pi}{2}\le \theta \le \frac{\pi}{2}.
\]
In this way, the kinetic-based flux can be evaluated as
\[
  \int_{-\frac{\pi}{2}}^{\frac{\pi}{2}} f(\theta) e^{ik\theta}d\theta
  = -\int_{\frac{\pi}{2}}^{\frac{3\pi}{2}} f(\theta) e^{ik(\pi-\theta)}d\theta
  = - e^{ik\pi} \overline{\int_{\frac{\pi}{2}}^{\frac{3\pi}{2}} f(\theta) e^{-ik\theta} d\theta}.
\]

\section{Numerical results} \label{sec:num}
This section is devoted to a series of numerical tests including spatially homogeneous,1D and 2D cases. Comparison with the analytical solution, the macroscopic limit and the microscopic particle simulation is provided to show the validity of our Poisson-EQMOM.
In the subsequent subsections, we set $\nu=1$ and $\mathtt d=0.2$ in Eq.(\ref{eq:vic}) for all cases.
Overall, it is demonstrated that the Poisson-EQMOM produces reasonable results in all cases.

\subsection{Spatially homogeneous case}
Inspired by \cite{GAMBA2015}, we first validate our computational model by a spatially homogeneous case $\partial_t f = Q(f)$ to test the convergence of the numerical results towards the known equilibrium $M_{\bar\theta}(\theta)$ of $Q(f)$, a von Mises distribution in Eq.(\ref{eq:vmeq}).
For such case, the initial distribution is specified as
\[
    f_0(\theta) = (1 + \cos(4 \theta)) \cdot \exp\left( -\cos\left(\pi \left(\frac{\theta}{2\pi} + \left(\frac{\theta}{2\pi}\right)^ 4\right)\right)\right)
\]
for $\theta \in [-\pi,\pi)$
with its moments $m_k(t=0)$ approximated by
\[
  m_k(0) = \sum_{j=0}^{1023} f_0(-\pi+j\Delta \theta) e^{ikj\Delta\theta}\Delta\theta, \quad \Delta\theta=\frac{2\pi}{1024}.
\]

Using the schemes in subsection \ref{subsec:col} and a CFL number of 0.5, we simulate the time evolution of $f(\theta)$ with the Poisson-EQMOM.
It is worth mentioning that the use of moment methods preserves mass conservation naturally.
As plotted in Fig.~\ref{fig:0D}(A), after a sufficiently long duration time ($t=20$), the reconstructed distribution $f$ is in good agreement with the von Mises distribution peaked at $\bar\theta=1.3477\pi$.
Fig.~\ref{fig:0D}(B) quantifies the time-evolved $L^2$-errors $\|f(\theta)-M_{\bar\theta}(\theta)\|_{L^2(-\pi,\pi)}$ and Fig.~\ref{fig:0D}(C) shows the evolution of the relative error of moments $|\bm m[f] - \bm m[M_{\bar{\theta}}]|/|\bm m[M_{\bar{\theta}}]|$. Here we denote $\bm m[f] = (m_0,m_1,...,m_N)$ with $m_k = m_k[f] = \int e^{ik \theta}f(\theta) d\theta$ for $k=0,1,...,N$.
We see exponential convergence at the initial stage. With an increase of $N$ from 4 to 32, the two errors reduce substantially. Tabel~\ref{tab:0D_time_cost} further shows the runtime to solve this case with different values of $N$.

\begin{figure}[htbp]
  \centering
  \subfloat[$N=8$]{\includegraphics[width=0.33\linewidth]{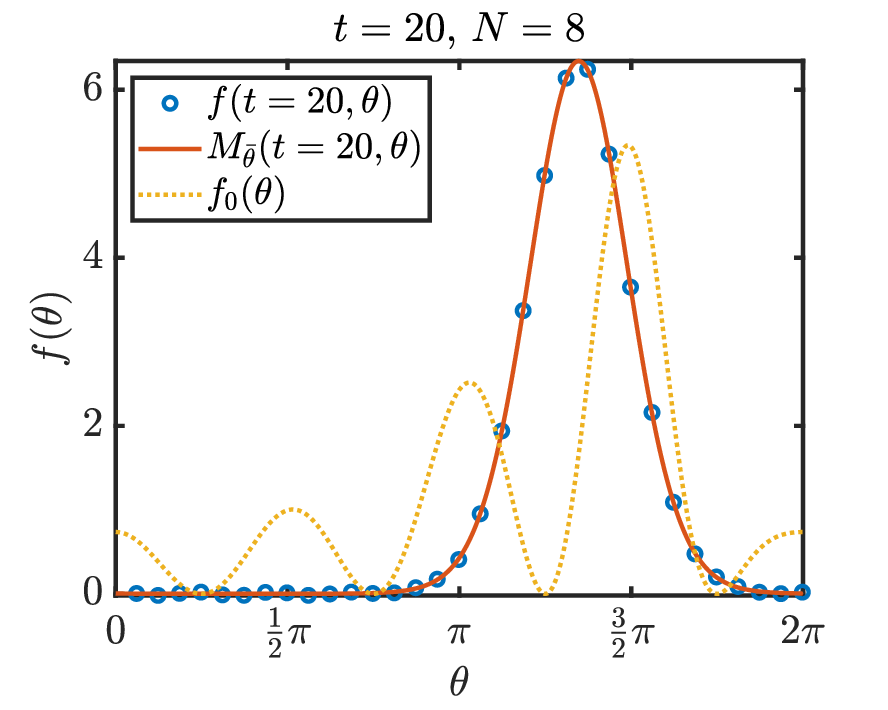}}
  \subfloat[$L^2$-errors]{\includegraphics[width=0.33\linewidth]{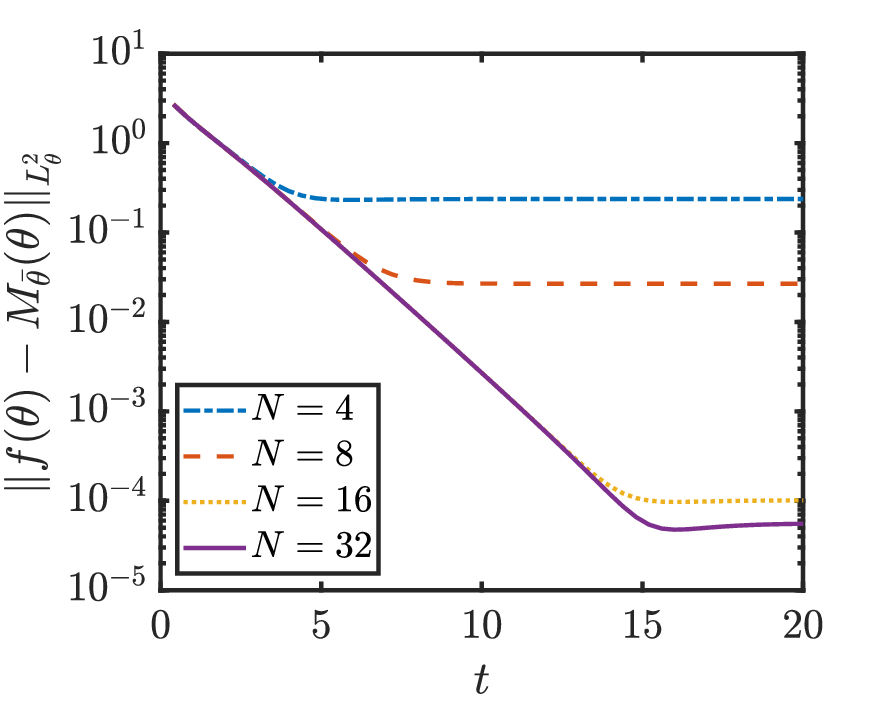}}
  \subfloat[Relative errors of moments]{\includegraphics[width=0.33\linewidth]{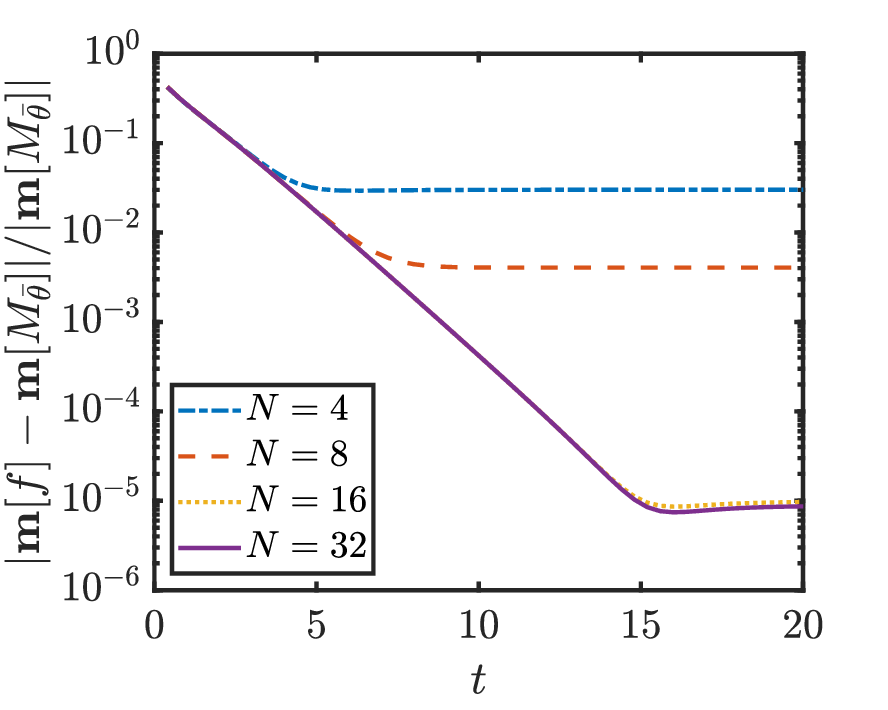}}
  \caption{The spatially homogeneous case.}
  \label{fig:0D}
\end{figure}

\begin{table}[htbp]
    \caption{The runtime to solve the spatially homogeneous case 100 times with different $N$. The results are obtained by the MATLAB code running on an Intel i7-1065G7 CPU.}
    \label{tab:0D_time_cost}
    \centering
    \begin{tabular}{c|cccc}
        \hline
        $N$ & 4 & 8 & 16 & 32 \\
        \hline
        Runtime (s) & 1.5151 & 2.1530 & 5.0664
& 21.1626 \\
        \hline
    \end{tabular}
\end{table}

\subsection{1D Riemann problems} \label{subsec:1DRie}
Our second case is the spatially one-dimensional Riemann problems with the initial conditions characterized by the density $\rho=m_0$ and the macroscopic velocity direction $\bar\theta = \arg m_1$.
As demonstrated in \cite{GAMBA2015}, the 1D problems provide ideal situations to test the performance of numerical methods in the hydrodynamic regime where the scaling factor $\epsilon$ in Eq.(\ref{eq:vic}) approaches zero.
In such a scenario, both the density and velocity direction obtained with the kinetic equation are expected to converge, as $\epsilon \to 0$, to those solved from the hydrodynamic model.

In the simulation, the computational domain $x\in [-5,5]$ is divided into 1000 uniform grid cells.
For the initial conditions, we set $(\rho_l,\bar\theta_l)$ for $x<0$ and  $(\rho_r,\bar\theta_r)$ for $x>0$.
Then the initial moments $m_k(0,x)$ are computed with the von Mises distribution determined by the macroscopic quantities. In all problems, we choose $N=12$ for the Poisson-EQMOM and the CFL number of 0.5. Simulations with $N=8$ and $16$ were also performed and the results (not shown) are quite close. The Neumann boundary condition is employed at both ends $x=\pm 5$.

Our initial data follow \cite{GAMBA2015}. The first case is given by
\begin{equation} \label{eq:1D1}
  (\rho_l,\bar\theta_l)=(2, 1.7), \quad
  (\rho_r,\bar\theta_r)=(0.218,0.5),
\end{equation}
in which a rarefaction wave forms in the hydrodynamic model.
The second case generates a shock wave and the initial conditions read as
\begin{equation} \label{eq:1D2}
  (\rho_l,\bar\theta_l)=(1, 1.5), \quad
  (\rho_r,\bar\theta_r)=(2,1.83).
\end{equation}
The third one is a contact discontinuity:
\begin{equation} \label{eq:1D3}
  (\rho_l,\bar\theta_l)=(1, 1), \quad
  (\rho_r,\bar\theta_r)=(1,-1).
\end{equation}

The simulated profiles of $\rho$ and $\bar\theta$ at $t=4$ are presented in Figs.~\ref{fig:1D1}-\ref{fig:1D3} for three values of $\epsilon$. Also shown in each plot is the solution of the hydrodynamic model extracted from \cite{GAMBA2015}.
It is clearly seen that with $\epsilon$ approaching 0, the Poisson-EQMOM predicted $\rho$ and $\bar\theta$ become quite close to the macroscopic solutions in all cases, well capturing the drastic change across the wave fronts. This serves as convincing evidence of the reliability of our computational model. Table.~\ref{tab:time_1D_case3} shows the runtime to complete the simulations for the third case with different values of $N$ and $\epsilon$.

\begin{figure}[htbp]
  \centering
  \subfloat
  {\includegraphics[height=4.5cm]{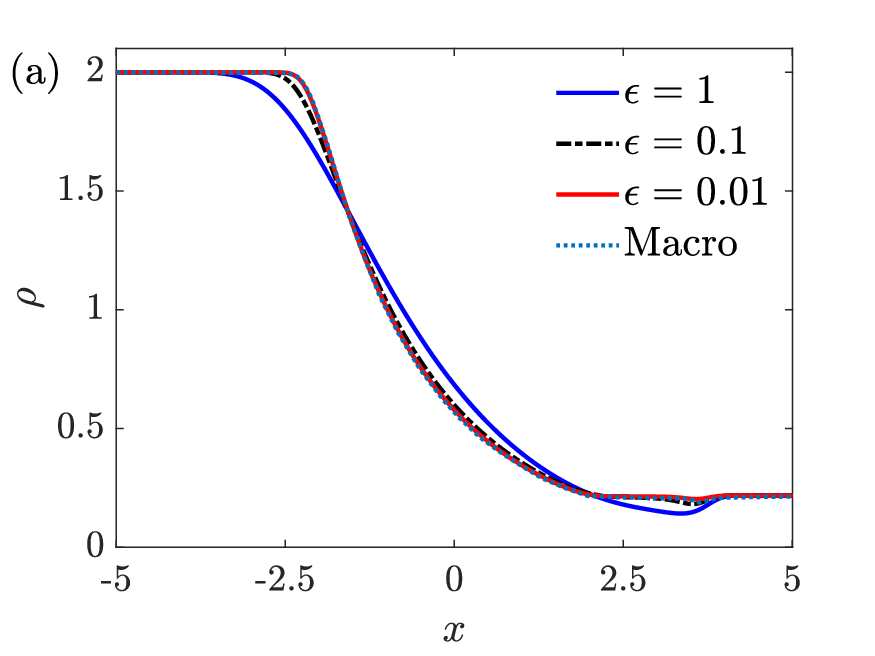}}
  \subfloat
  {\includegraphics[height=4.5cm]{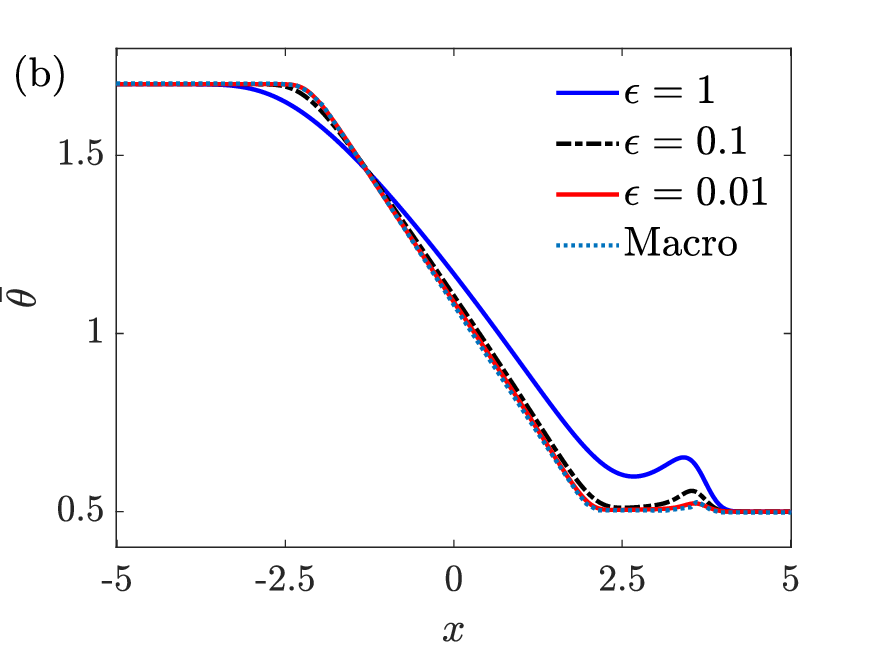}}
  \caption{The 1D Riemann problem (\ref{eq:1D1}) with a rarefaction wave. The profiles of density $\rho$ and velocity direction $\bar\theta$ at $t=4$. The `macro' profiles are extracted from \cite{GAMBA2015}.}
  \label{fig:1D1}
\end{figure}

\begin{figure}[htbp]
  \centering
  \subfloat
  {\includegraphics[height=4.5cm]{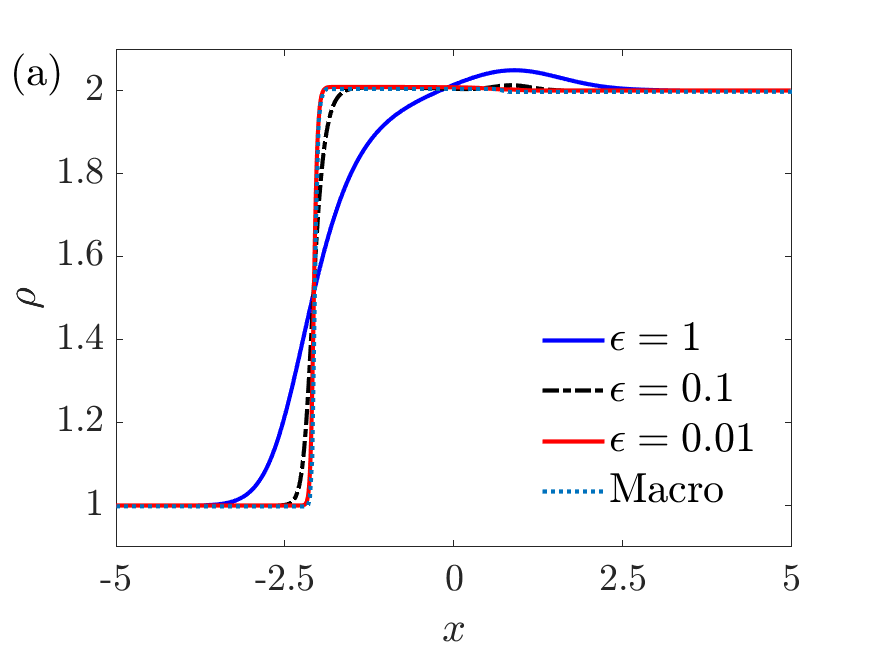}}
  \subfloat
  {\includegraphics[height=4.5cm]{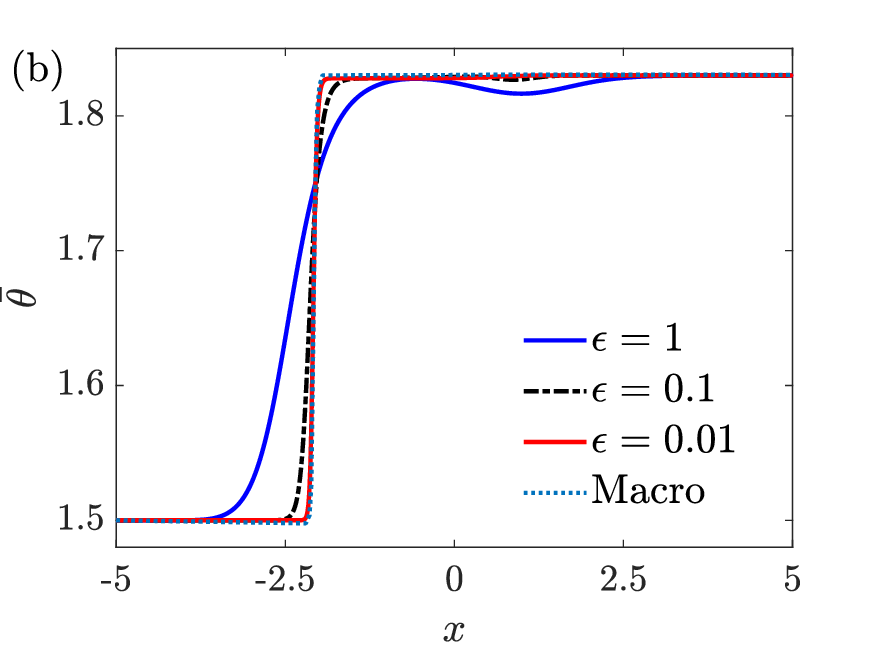}}
  \caption{The 1D Riemann problem (\ref{eq:1D2}) with a shock wave. The profiles of density $\rho$ and velocity direction $\bar\theta$ at $t=4$. The `macro' profiles are extracted from \cite{GAMBA2015}.}
  \label{fig:1D2}
\end{figure}

\begin{figure}[htbp]
  \centering
  \subfloat
  {\includegraphics[height=4.5cm]{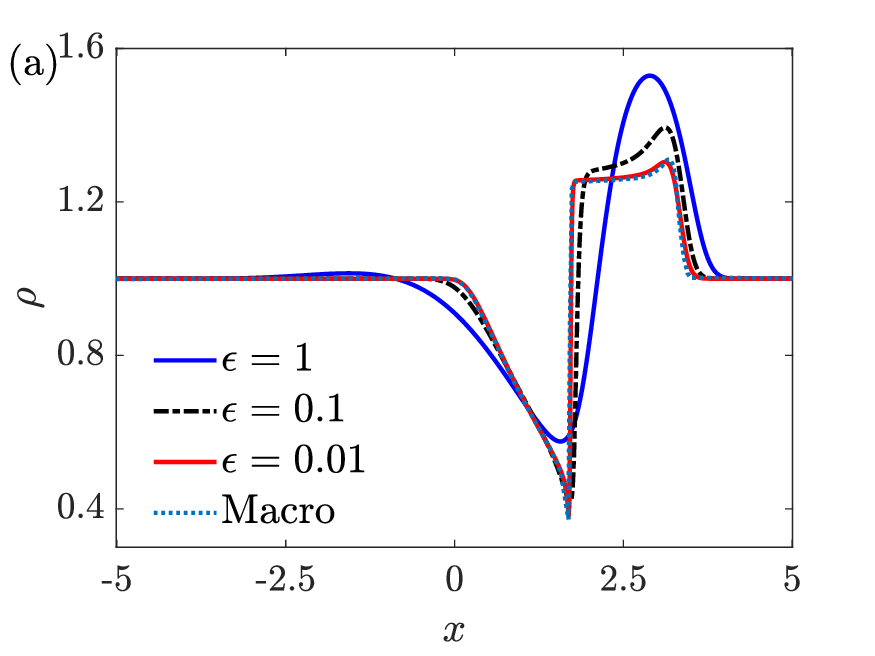}}
  \subfloat
  {\includegraphics[height=4.5cm]{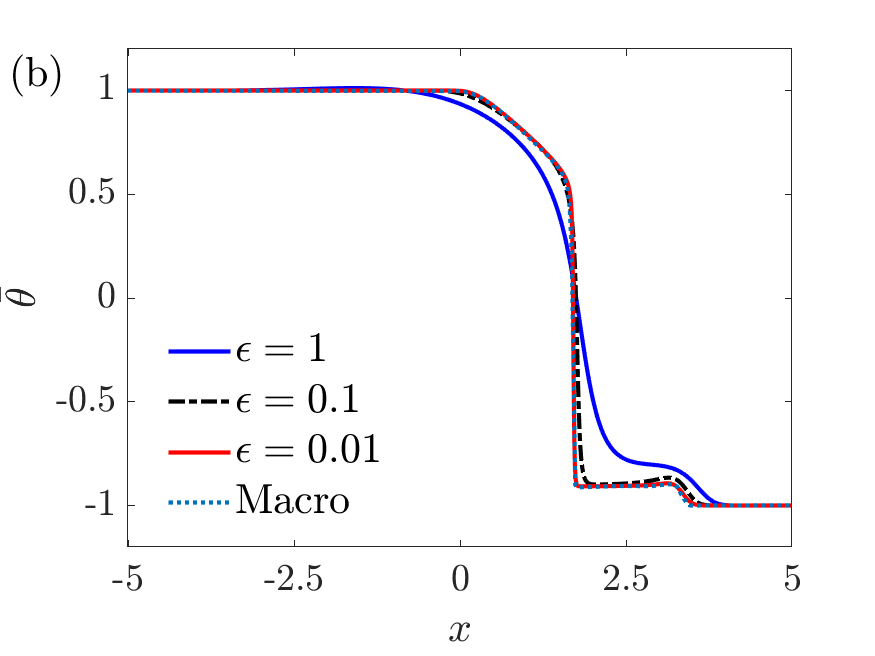}}
  \caption{The 1D Riemann problem (\ref{eq:1D3}) with a contact discontinuity. The profiles of density $\rho$ and velocity direction $\bar\theta$ at $t=4$. The `macro' profiles are extracted from \cite{GAMBA2015}.}
  \label{fig:1D3}
\end{figure}

\begin{table}[htbp]
    \caption{The runtime for the third case (\ref{eq:1D3}) in section \ref{subsec:1DRie} with different values of $N$ and $\epsilon$. The results are obtained by the MATLAB code running on an Intel i7-1065G7 CPU.}
    \label{tab:time_1D_case3}
    \centering
    \begin{tabular}{c|ccc}
         \hline
         Runtime (s) & $\epsilon=1$ & $\epsilon=0.1$ & $\epsilon=0.01$ \\
         \hline
         $N=8$ & 670.84 & 657.58 & 819.10 \\
         $N=12$ & 1384.20 & 1306.16 & 1584.60 \\
         $N=16$ & 2209.37 & 2215.02 & 2735.97 \\
         \hline
    \end{tabular}
\end{table}

\subsection{Spatially 2D cases}
In this subsection, several spatially 2D simulations are performed on a physical domain $(x,y)\in[-5,5]\times [-5,5]$ employing different boundary conditions. The domain is divided into $50\times 50$ uniform grids, resulting in $\Delta x=\Delta y=0.2$.
We set $\epsilon=1$ in Eq.(\ref{eq:vic}) for comparisons with the microscopic simulation results.

The initial condition is generated with a constant density $m_{0,ij}(t=0)=1$, and the higher-order moments are determined by the von Mises distributions with the mean velocity direction $\bar\theta_{ij}$ uniformly distributed on $[-\pi,\pi]$ for all $1\le i,j\le 50$.
Our simulation marches in time to examine specific long-time patterns.
Throughout this subsection we choose $N=4, 8, 16$ for the Poisson-EQMOM and the CFL number of 0.25. The runtime ranges from several hours to days depending on the choice of $N$.

For a comparison, the microscopic Vicsek model is simulated on the same domain using 10 000 particles. A continuous version of the Vicsek model is applied here \cite{degond2008,GAMBA2015}. The particle speed is 1. The radius of interaction is 0.2 and the discrete time step is taken as 0.01. The intensity of noise $\mathtt d=0.2$. Similarly, at the initial state, the particles are uniformly distributed in space and velocity direction.
If the wall condition is present, each impacting particle reflects with its normal velocity component reversed.

\subsubsection{Periodic boundaries}
We first consider the case with four periodic boundary conditions. In this case, for any function $\phi=\phi(t,x,y)$, we have
\[
  \phi(t,x,y)=\phi(t,x+10p,y+10q)
\]
for any $p,q\in\mathbb Z$. Fig.~\ref{fig:2Dp} (left) shows the typical density and velocity distributions at $t=200$ predicted by the 16-node Poisson-EQMOM, whereas the right panel exhibits the coarse-grained results of the particle model averaged from $t=200$ to $t=300$ (to reduce the fluctuation caused by a finite number of particles).
Besides, the $L^2$ errors between the density (or velocity) and the spatially mean value are illustrated in Fig.~(\ref{fig:2D_periodic_errors}) for $N=4,8,16$. It is clearly seen that, under such a noise strength ($\mathtt d=0.2$), the long-time behavior is a global alignment state with nearly homogeneous density in space, while the velocity direction is random since the initial velocity direction is randomly distributed for both the Poisson-EQMOM and the particle simulation.
A density oscillation could be observed initially and wanes effectively over time. As indicated in Fig.~(\ref{fig:2D_periodic_errors}), the cases with $N=8$ and $N=16$ have similar accuracy while both converge to the alignment state faster than that with $N=4$.
Hence, the Poisson-EQMOM gives results consistent with the particle model, and $N=16$ seems to be sufficient for a quantitatively reasonable simulation at the kinetic-macroscopic scales.

\begin{figure}[htbp]
\centering
\subfloat
{\includegraphics[width=0.4\linewidth]{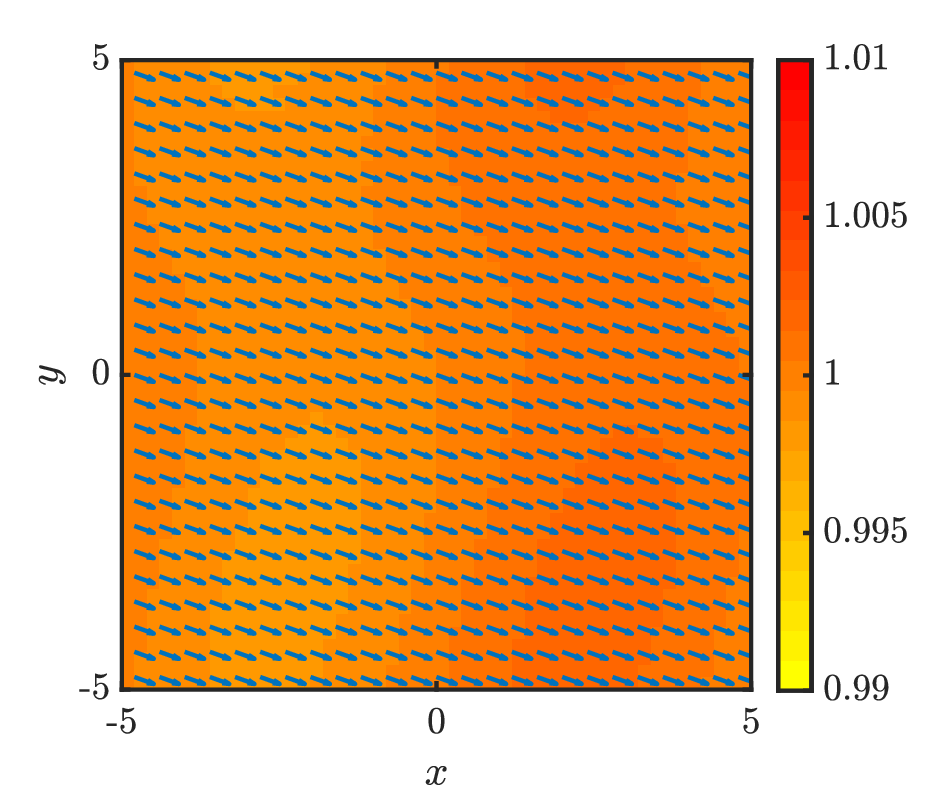}}
\subfloat
{\includegraphics[width=0.4\linewidth]{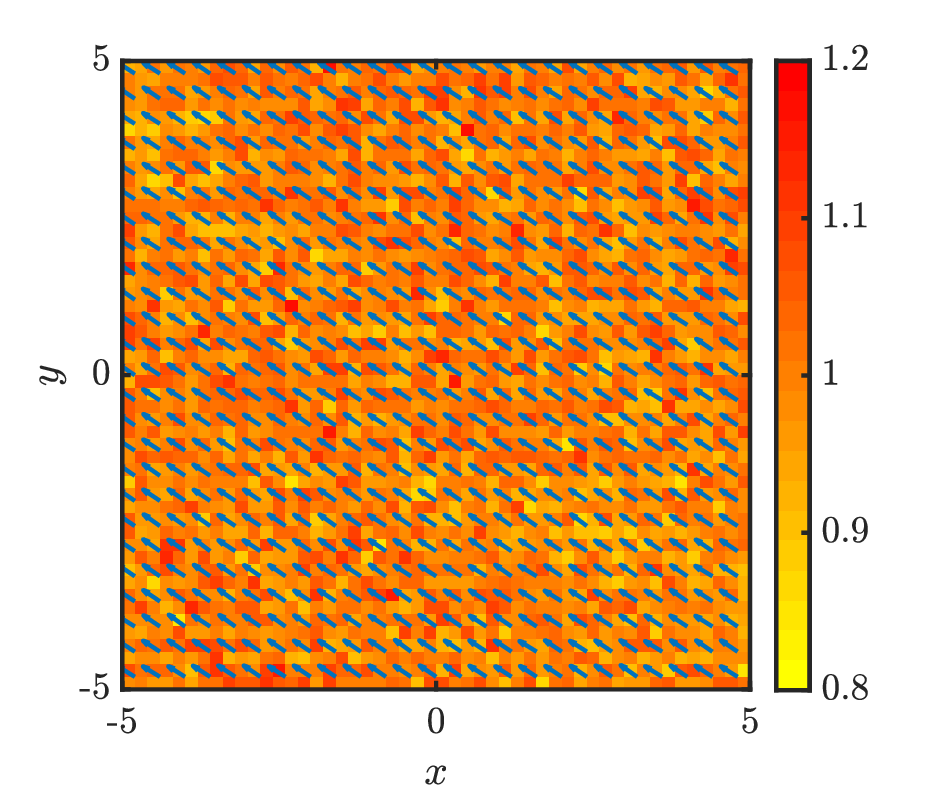}}
\caption{Density and velocity distributions at $t=200$ by the 16-node Poisson-EQMOM (Left) and the Vicsek model (Right) for the 2D case with periodic boundary conditions. The right image is obtained by taking average from $t=200$ to $t=300$ to reduce the fluctuation of the particle model.}
\label{fig:2Dp}
\end{figure}

\begin{figure}[htbp]
\centering
\subfloat
{\includegraphics[width=0.4\linewidth]{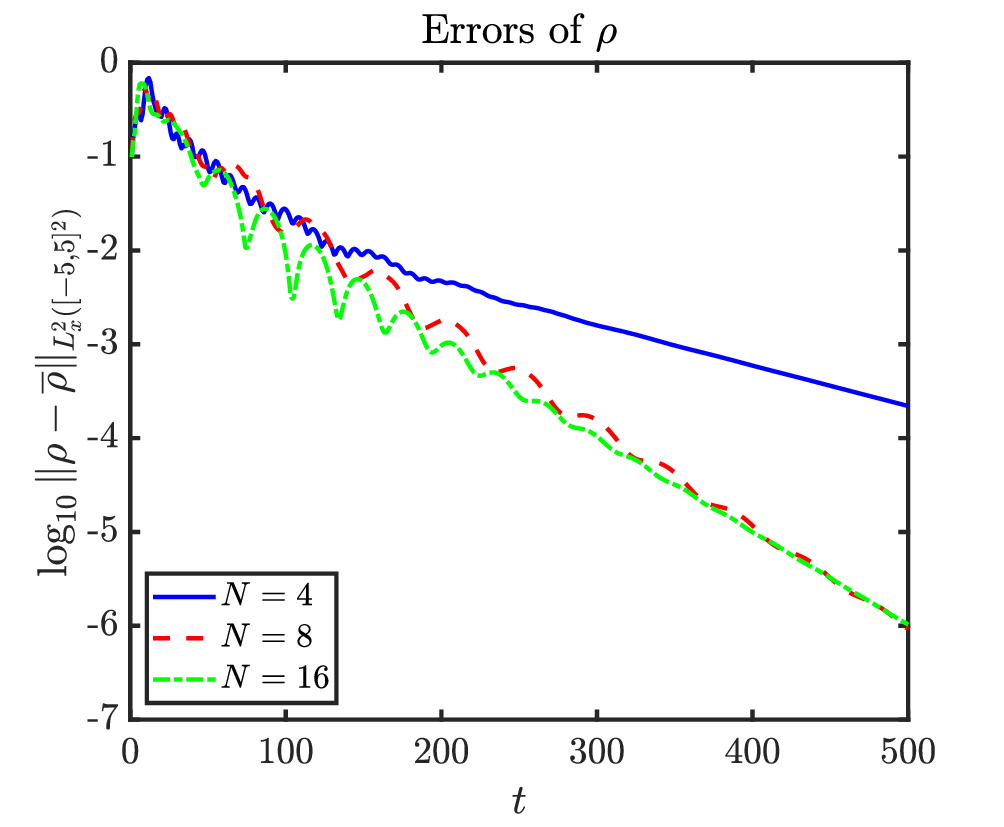}}
\subfloat
{\includegraphics[width=0.4\linewidth]{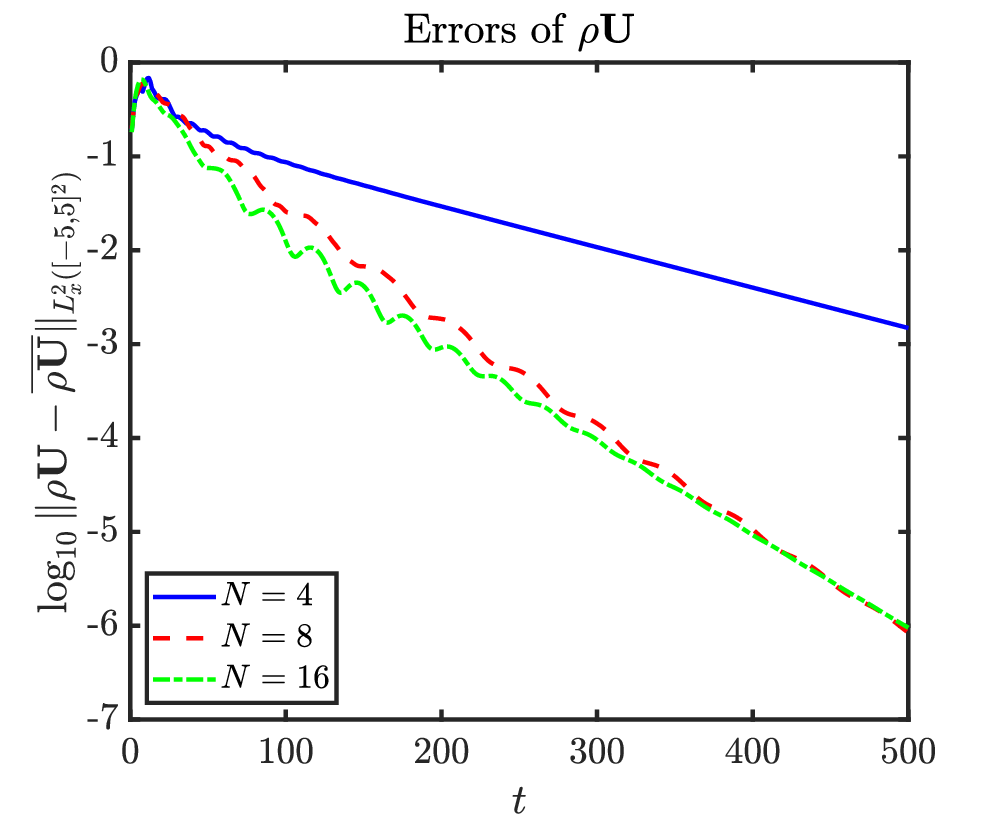}}
\caption{$L^2$ errors of density and velocity by the Poisson-EQMOM with different values of $N$ for the 2D case with periodic boundary conditions.}
\label{fig:2D_periodic_errors}
\end{figure}

\subsubsection{Two reflexive walls} \label{subsubsec:tworeflw}
Our second case sets the boundaries $x=\pm 5$ as reflexive walls, leaving $y=\pm 5$ still as periodic conditions. That is, $\phi(t,x,y)=\phi(t,x,y+10q)$ for any function $\phi$ and $q\in\mathbb Z$.
The time evolution of the density and velocity distributions is plotted in Fig.~\ref{fig:2D2} for different snapshots, showing a more complex pattern than those of the previous case with all periodic conditions.
After a transient period, the system is observed to have converged to a density-oscillating state between the opposite walls with a full cycle lasting about 40 unit time, as revealed at both the kinetic (Figs.~\ref{fig:2D2}a,c,e) and microscopic (Figs.~\ref{fig:2D2}b,d,f) levels. The snapshots in Fig.~\ref{fig:2D2} cover half of the cycle: the bulk density is moving to the left, generating a huge peak adjacent to the left wall in Figs.~\ref{fig:2D2}c\&d, and then goes back towards the right wall.
Fig.~\ref{fig:6_3_2_period_statistic} quantitatively illustrates the local density evolution at $\bm x = (-5, 5)$ predicted by different methods. Clearly, all results reveal `density waves' with nearly the same period and amplitude except for the differences caused by the random initial data.
The consistency of the micro- and meso-scales again justifies our use of the Poisson-EQMOM.

It is interesting to point out some features of the system.
In the near-wall region, the particles move almost along the wall because the normal velocity components are eliminated due to particle rebound and the alignment interaction. By contrast, particles migrate along the $x$-axis far away from the wall, transporting the `density wave front' back and forth at a mean `bulk' velocity about 0.5. The crowded area can be 5-7 times denser than the most sparse region (see Fig.~\ref{fig:6_3_2_period_statistic}).
However, it is postulated that this density inhomogeneity may be distinct from the `phase separation' state \cite{chate2020} which is believed to exist at larger noise strengths. The kinetic simulation implies that the low-density regions here are still aligned effectively and are thus not suitable to be regarded as a disordered state (with zero macroscopic speed). A thorough investigation is beyond our current scope and will be left for future work.

\begin{figure}[htbp]
  \centering
  \subfloat
  {\includegraphics[height=4.8cm]{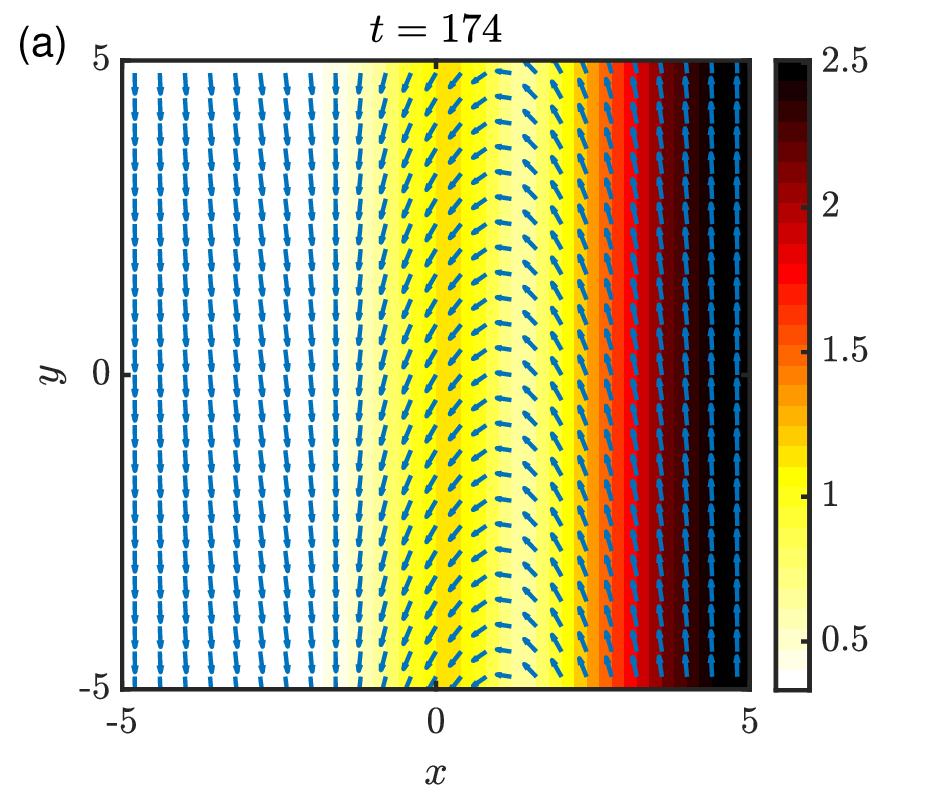}}
  \subfloat
  {\includegraphics[height=4.8cm]{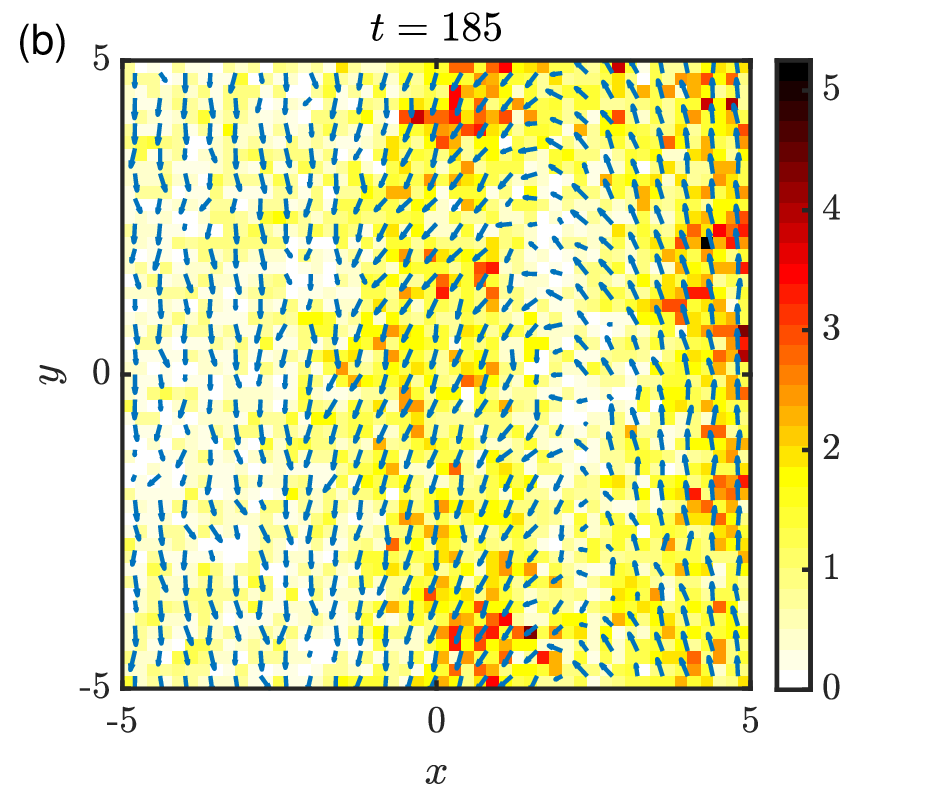}}
  \\
  \subfloat
  {\includegraphics[height=4.8cm]{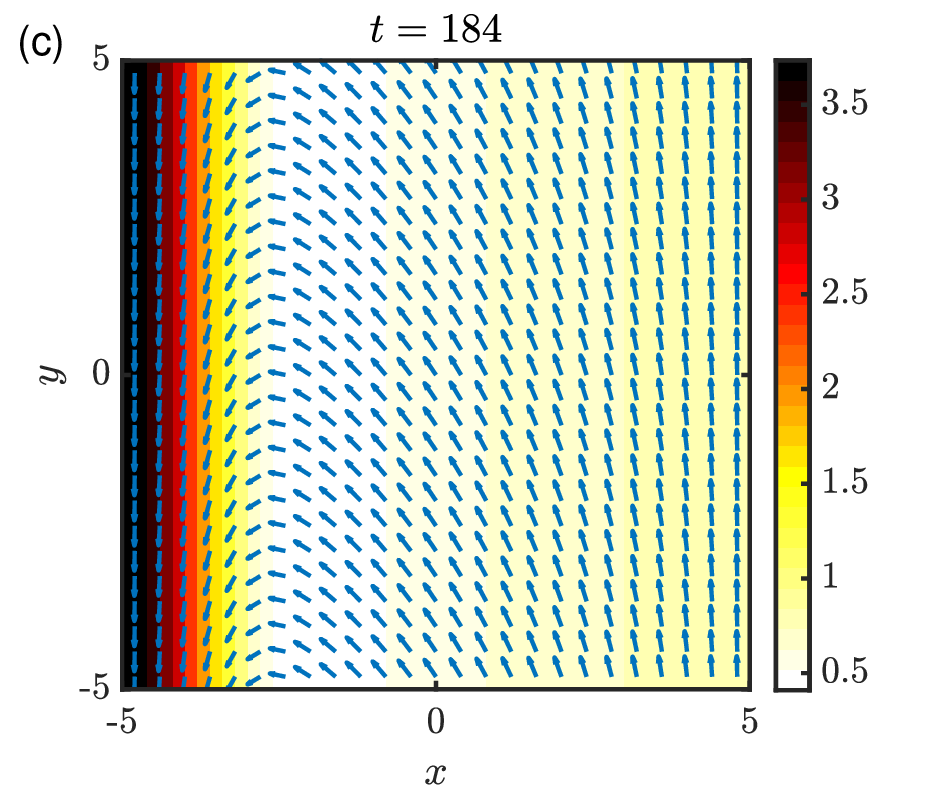}}
  \subfloat
  {\includegraphics[height=4.8cm]{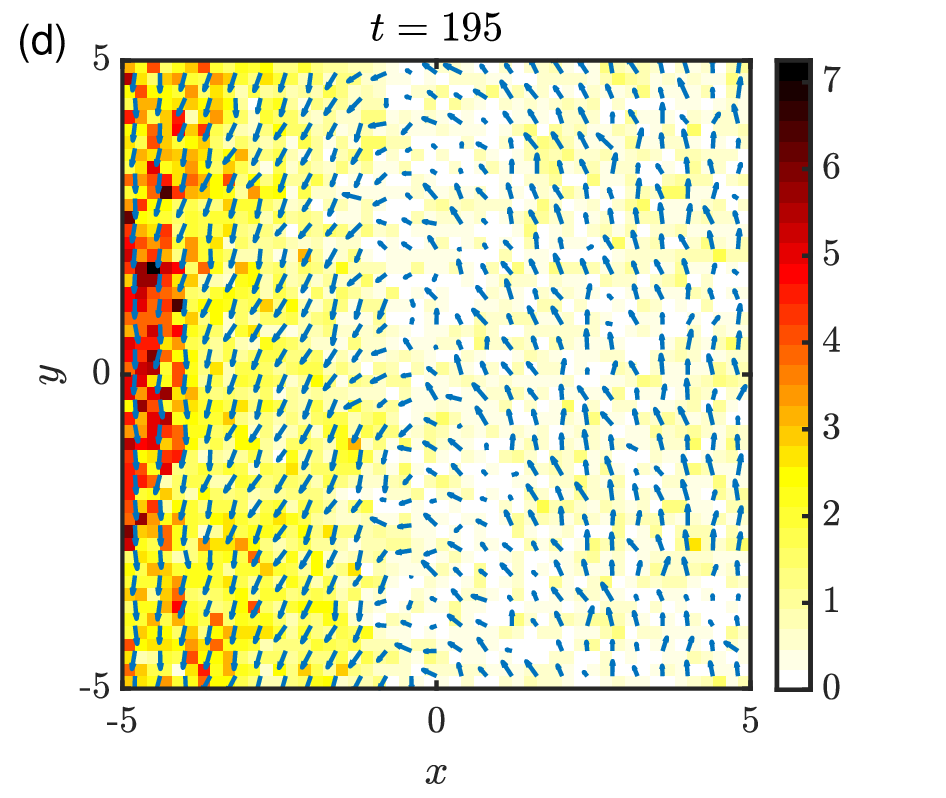}}
  \\
  \subfloat
  {\includegraphics[height=4.8cm]{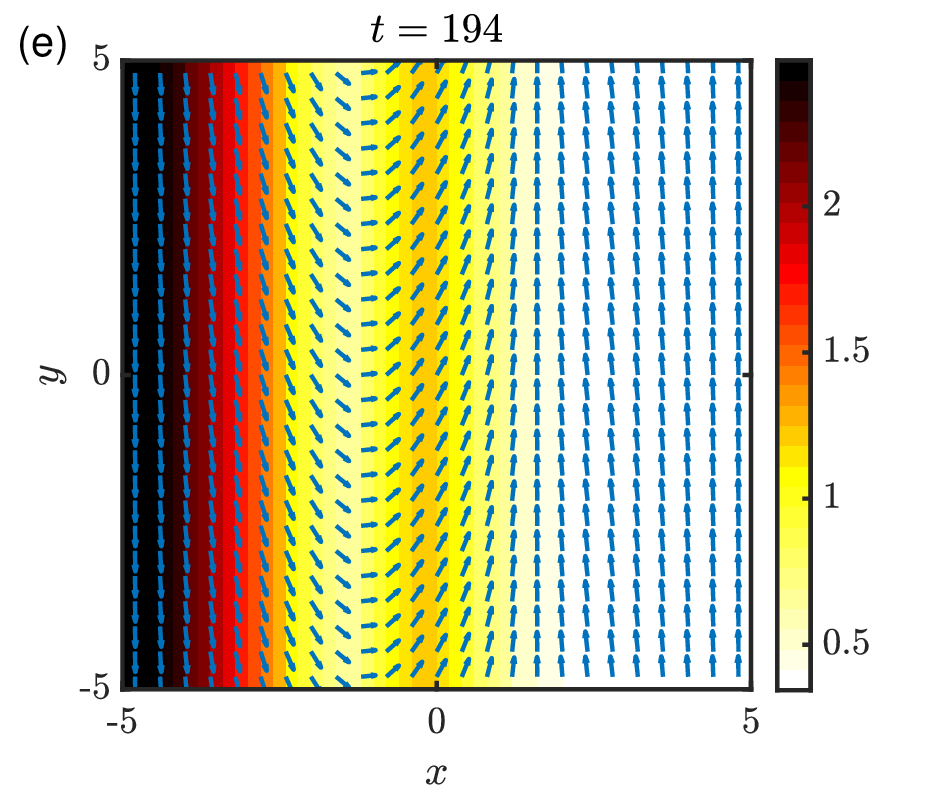}}
  \subfloat
  {\includegraphics[height=4.8cm]{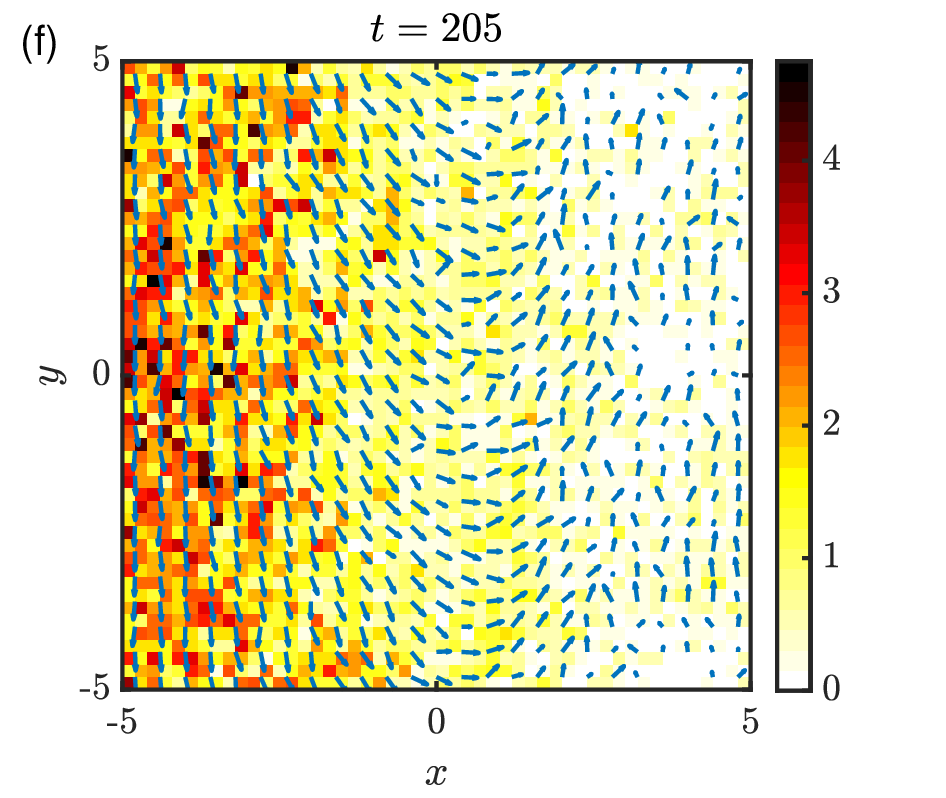}}
  \caption{Density and velocity distributions at different snapshots for the 2D case with two reflexive walls located at $x=\pm 5$. Periodic conditions are imposed on $y=\pm 5$.
  Left column: 16-node Posson-EQMOM;
  Right column: Vicsek model.
  The snapshots cover half of the full density-oscillating cycle.}
  \label{fig:2D2}
\end{figure}

\begin{figure}
    \centering
    \includegraphics[width=0.9\linewidth]{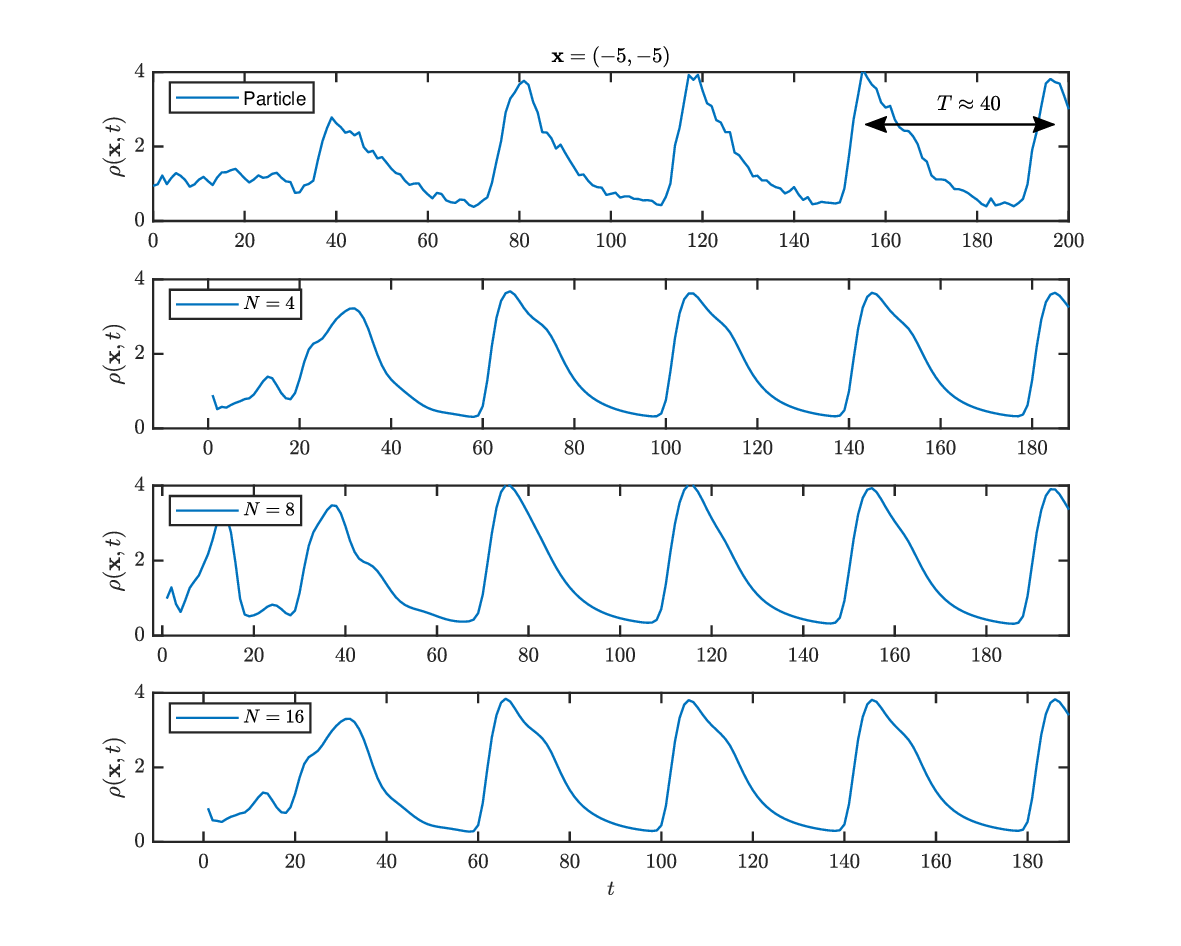}
    \caption{Density oscillations $\rho(\bm x,t)$ for the case in subsection \ref{subsubsec:tworeflw}. The particle result (top panel) is obtained by taking the average on $x=-5$. The latter three results are from the Poisson-EQMOM with different values of $N$ and $\bm x=(-5,-5)$.}
    \label{fig:6_3_2_period_statistic}
\end{figure}

\subsubsection{Four walls with vortex formation}
\label{subsubsec:fourwall}
Having witnessed the intriguing pattern induced by walls, we further assume in this case that the domain is entirely bounded by reflexive walls.
As shown in Fig.~\ref{fig:2D3}, after a transient period, both the kinetic and particle simulations predict a stable bulk vortex rotating around the origin with the particle density concentrating adjacent to the walls. As a result, the center is nearly empty. The same pattern was reported in \cite{GAMBA2015}.
The vortex forms because each wall drives the particles to stream along. Meanwhile, there is a lack of mechanism to effectively transport particles away from the wall, leading to mass accumulation near the wall and a central void, which in turn strengthen the alignment in the `tangential' direction. We can tell from the velocity field that the particles are locally aligned while having a global flow structure.
Fig.~\ref{fig:6_3_3_period} gives the local density evolution at $\bm x = (-5, 5)$ predicted by the Poisson-EQMOM with different values of $N$. The period of rotation is around 30-40 unit time.

\begin{figure}[htbp]
  \centering
  \subfloat
  {\includegraphics[height=4.8cm]{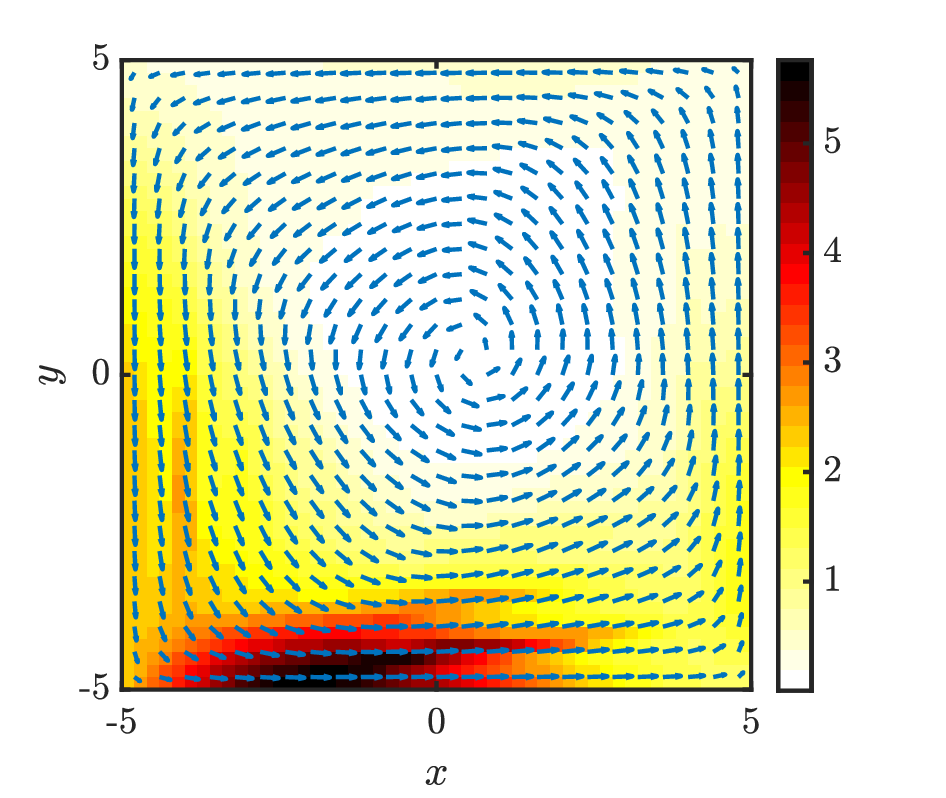}}
  \subfloat
  {\includegraphics[height=4.8cm]{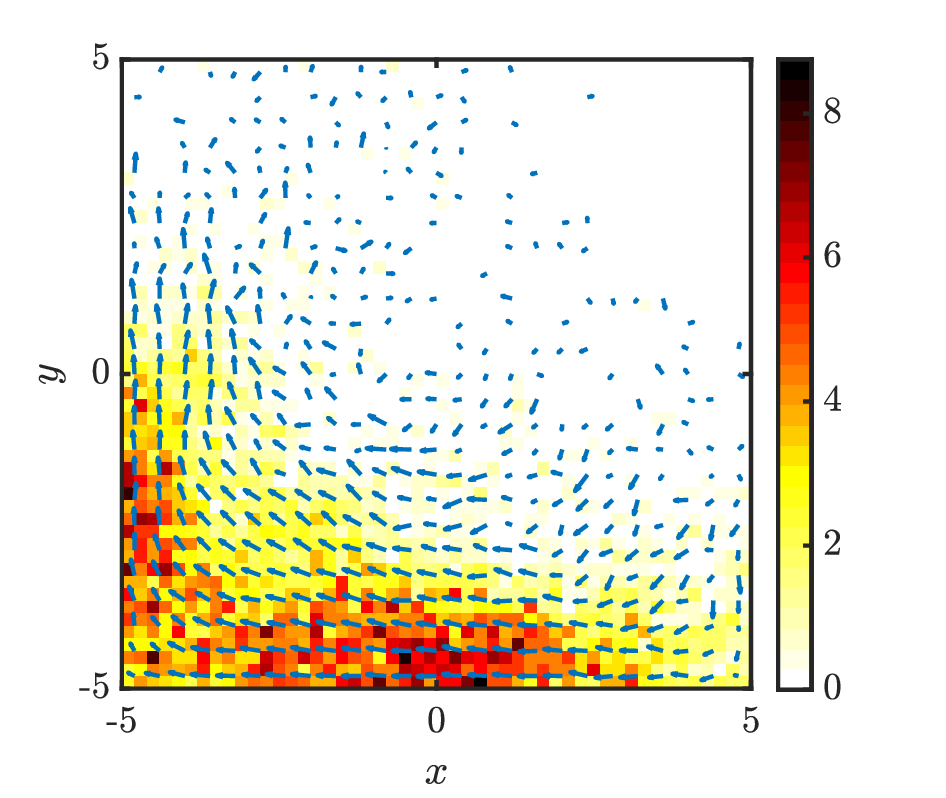}}
  \caption{Density and velocity distributions at $t=200$ for the 2D case with four reflexive walls: (Left) 16-node Poisson-EQMOM; (Right) Vicsek model.}
  \label{fig:2D3}
\end{figure}

\begin{figure}
    \centering    \includegraphics[width=0.9\linewidth]{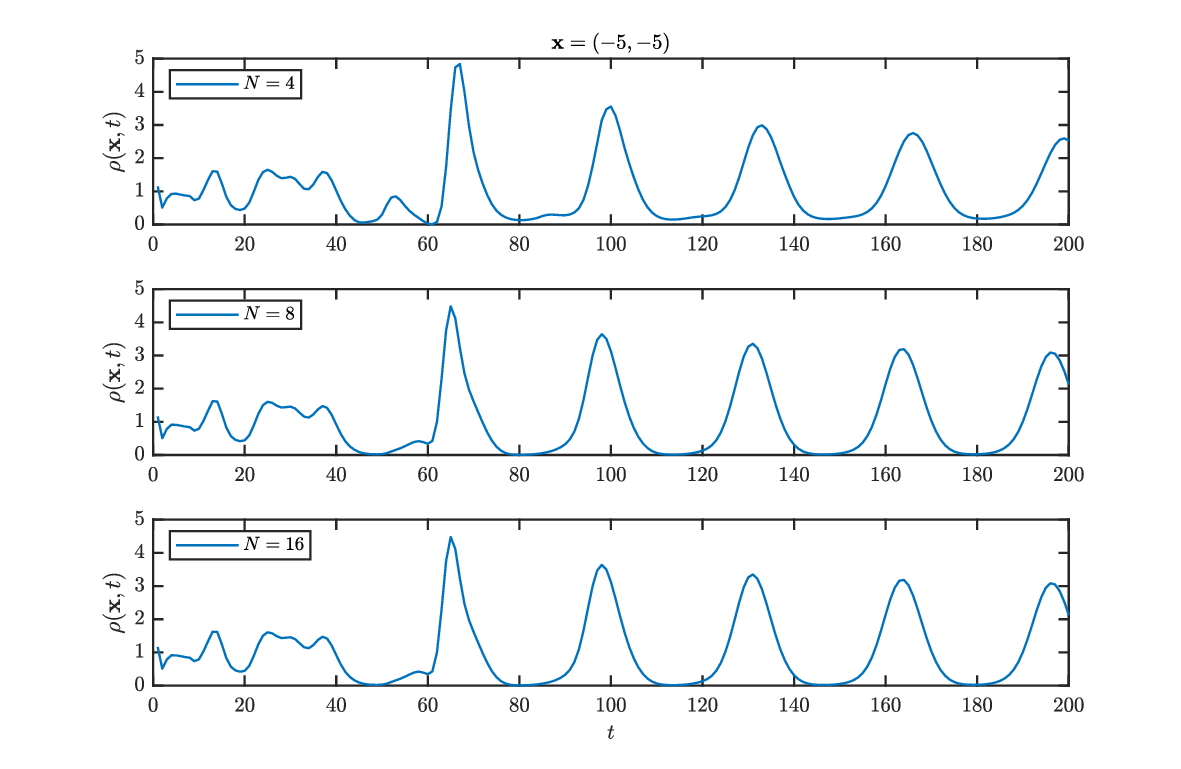}
    \caption{Density oscillations $\rho(\bm x,t)$ for the case in subsection \ref{subsubsec:fourwall}. The results are from the Poisson-EQMOM with different values of $N$ and $\bm x=(-5,-5)$.}
    \label{fig:6_3_3_period}
\end{figure}

\section{Conclusions} \label{sec:con}

This paper develops a novel moment method, called the Poisson-EQMOM, for 2D kinetic equations with velocity of constant magnitude (KE-VC). The idea is to approximate the distribution function with a convex combination of finite shifted homoscedastic Poisson kernels.
We then apply the Poisson-EQMOM to a KE-VC associated with the seminal Vicsek model to generate a new moment closure system.

Besides being \textit{positivity-preserving} and \textit{in conservative form}, the Poisson-EQMOM is shown to be well defined for \textit{all} physically relevant moments. More importantly, as long as the original distribution has a positive lower bound, the resultant approximation \textit{converges} as the number of moments goes to infinity.
The proofs are quite technical and crucially rely on the choice of Poisson kernel.
Based on these nice properties, we devise a robust moment inversion algorithm.
This algorithm involves a `lifting' technique which efficiently enhances the approximation precision while reducing the computational time (as compared with no lifting). The centers of the Poisson modes are computed with the aid of the orthogonal polynomial theory on the unit circle. The inversion process is shown to be stable for all our practical simulations with moderate numbers of moments.

The Poisson-EQMOM is numerically validated with the moment system based on the Vicsek model, where a delicate treatment of the kinetic-based flux is involved.
For all test cases ranging from spatially homogeneous to 2D, the numerical results are satisfactory.
Precisely, all the numerical results indicate the convergence of our method, well capture the flow discontinuities at the macroscopic level, and reveal fascinating wall-induced patterns (of density waves and bulk vortex) generally in consistent with the microscopic simulations.

As the Poisson-EQMOM is a new method, there are many directions for further research, including: (\textit{i}) exploring its extension to three dimensions; (\textit{ii}) proving hyperbolicity in general cases; (\textit{iii}) analyzing the asymptotic-preserving properties of the numerical schemes; (\textit{iv}) conducting a deeper analysis of the wall-induced density band structures observed in this paper and applying the Poisson-EQMOM to more systems.

\section*{Data availability}
All our simulations are implemented in MATLAB. Our codes are now open-sourced at Github (\href{https://github.com/CheYHinSpark/Poisson-EQMOM}{https://github.com/CheYHinSpark/Poisson-EQMOM}).

\section*{Acknowledgments}
The authors owe their sincere appreciation to Prof. Hui Yu at Xiangtan University for her kind assistance in particle simulations of the Vicsek model. The authors are also grateful to Prof. Pierre Degond at CNRS, Prof. Remi Abgrall at UZH, Prof. Shuiqing Li and Mr. Hanlin Lin at Tsinghua University, and Prof. Xianmin Xu at CAS for the insightful discussions.

\appendix

\section{}
\begin{proof}[Proof of Theorem \ref{thm:lift}]
  For both $f$ and $f_N$, recall the Fourier series
  \[
    f = \frac{1}{2\pi} \sum_{k\in\mathbb Z} \hat f(k) e^{ik\theta}, \quad f_N= \frac{1}{2\pi} \sum_{k\in\mathbb Z} \hat{f_N} (k) e^{ik\theta}
  \]
  with $\hat f(k)= \overline{m_k}$ for all $k\in\mathbb Z$ and $\hat{f_N}(k)=\hat f(k)$ for $|k|=0,\dots,N$.
  On the other hand, we see from Eqs.(\ref{eq:peans} \& \ref{eq:poisdef}) that
  \[
  \begin{aligned}
    f_N &= \left( \sum_{|k|\le N} + \sum_{|k|>N} \right) \frac{r_N^{|k|}}{2\pi} \left( \sum_{\alpha=1}^N \rho_{\alpha,N} e^{ik\phi_{\alpha,N}} \right) e^{-ik\theta} \\
    &= \sum_{|k|\le N} \hat f(k) e^{ik\theta} + \sum_{|k|>N} e^{ik\theta} r_N^{|k|} \cdot O(1),
  \end{aligned}
  \]
  in which the $N$-dependence of $\rho_{\alpha}$, $\phi_{\alpha}$ and $r$ is explicitly presented. The second term of the last expression contains $O(1)$ because $|\sum \rho_{\alpha} e^{ik\phi_{\alpha}}| \le \sum \rho_\alpha = m_0$.

  Then we have, for any $2\le p<\infty$, that
  \begin{equation} \label{eq:normineq}
  \begin{split}
    \| f_N - f \|_{L^p} &\le  \| f_N - \sum_{|k|\le N} \hat f(k) e^{ik\theta} \|_{L^p} + \| \sum_{|k|> N} \hat f(k) e^{ik\theta} \|_{L^p} \\
    &\le C \left( \sum_{k>N} r_N^{kp'} \right)^{\frac{1}{p'}} + \| \sum_{|k|> N} \hat f(k) e^{ik\theta} \|_{L^p}
  \end{split}
  \end{equation}
  with $p'=1/(1-p^{-1}) \in (1,2]$. The last step follows from the Hausdorff-Young's inequality.
  Hereinafter $C$ denotes a finite positive constant.

  We also need the fact that $\lim\limits_{N\to\infty} r_N$ exists because the sequence $\{r_N\} \subset [0,1]$ is increasing. To see this, notice that the minimum eigenvalue of $H_{N-1}(\bm M^*_N(r_{N-1}))$ is zero, implying $\lambda(r_{N-1};N)\le0$. Since $\lambda(r_N;N)=0$ and $\lambda(r;N)$ is strictly increasing on $r$ (Proposition \ref{prop:eiginc}), it is only possible that $r_N\ge r_{N-1}$.

  (\textbf{i}).
  Using $f'\in L^2(-\pi,\pi)$, we obtain from the Parseval's theorem that
  \[
    \infty > \int |f'|^2d\theta =
    \sum_{k\in\mathbb Z} |\hat{f'}(k)|^2 = \sum_{k\in\mathbb Z} k^2 |m_k|^2.
  \]
  As a consequence, the series $\sum\limits_{|k|\le N} \overline{m_k} e^{ik\theta} \to 2\pi f$ uniformly as $N\to\infty$ because
  \[
    | \sum_k \overline{m_k} e^{ik\theta} |^2 \le \left( \sum_k |k \overline{m_k} e^{ik\theta}|^2 \right) \left( \sum_k k^{-2} \right) < \infty.
  \]
  This fact immediately leads to $\| \sum\limits_{|k|>N} \hat f(k) e^{ik\theta} \|_{L^2} \to 0$ as $N\to \infty$. Thanks to Eq.(\ref{eq:normineq}), it then suffices to show
  \begin{equation} \label{eq:rseqvan}
    \sum_{k>N} r_N^{2k} = \frac{r_N^{2N+2}}{1-r_N^2} \le C \frac{r_N^{2N}}{1-r_N} \to 0 \quad \text{as} \quad N\to\infty,
  \end{equation}
  which holds trivially if $\lim r_N <1$. As such, in what follows we shall assume $\lim r_N=1$.

  As we need only to consider very large $N$, one can specify a sufficiently large $n\le N$ such that
  \begin{equation} \label{eq:nconstrain}
    \sum_{|k|< n} \overline{m_k} e^{ik\theta} \ge c\pi \quad \text{and}\quad
    \sum_{|k|\ge n}k^2 |m_k|^2 < \delta_n \to 0^+,
  \end{equation}
  which are both consequences of $f'\in L^2$.

  Take $a=(a_0,\dots,a_N)\in\mathbb C^{N+1}$ and define
  \begin{equation} \label{eq:qnra}
  \begin{split}
    q_N(r;a):&=\bar a^T H_N(\bm M^*_N(r)) a
    = \sum_{k,l=0}^N \overline{a_k}a_l m_{l-k} r^{-|l-k|} \\
    &= \frac{1}{2\pi} \int_{-\pi}^\pi | a(\theta) |^2 \left( \sum_{|k|\le N} r^{-|k|} \overline{m_k} e^{ik\theta} \right) d\theta
  \end{split}
  \end{equation}
  with $a(\theta) = \sum_l a_l e^{il\theta}$.
  The last expression can be verified by applying the orthogonality of the bases $e^{ik\theta}$.
  Since $H_N(\bm M^*_N(r_N))$ is singular, there exists $\tilde a\in\mathbb C^{N+1}$ with $\| \tilde a\|=1$ such that
  \[
    0 = |q_N(r_N;\tilde a)| = |A+B| \ge |A| - |B| .
  \]
  Here we denote
  \[
    A = \frac{1}{2\pi} \int_{-\pi}^\pi | \tilde a(\theta) |^2 \left( \sum_{|k|< n} r_N^{-|k|} \overline{m_k} e^{ik\theta} \right) d\theta.
  \]
  Recall $\lim r_N=1$. Take a sufficiently large $N$ with $r_N$ so close to 1 that the real finite sum $\sum\limits_{|k|< n} r_N^{-|k|} \overline{m_k}e^{ik\theta}$ is close to $\sum\limits_{|k|< n} \overline{m_k}e^{ik\theta}$.
  Considering Eq.(\ref{eq:nconstrain}), one may require $\sum\limits_{|k|< n} r_N^{-|k|} \overline{m_k}e^{ik\theta} \ge c\pi/2$ and hence $A \ge c\pi/2$ (because $\int |\tilde a(\theta)|^2 d\theta=2\pi$).

  For the rest part $B=q_N(r_N;\tilde a)-A$, we have
  \begin{equation} \label{eq:Bestim}
  \begin{split}
    |B|^2 &\le \left( \sum_{|k|=n}^N r_N^{-|k|} |m_k| \right)^2
    \le 4 \left( \sum_{k=n}^N k^{-2} r_N^{-2k} \right) \left( \sum_{k=n}^N k^2|m_k|^2 \right) \\
    &\le 8\delta_n \sum_{k=n}^N \frac{r_N^{-2k}}{k(k+1)}
    \le 8\delta_n \left( \frac{r_N^{-2n}}{n} + (1-r_N^2)r_N^{-2N}\sum_{k=n+1}^N \frac{r_N^{2N-2k}}{k} \right).
  \end{split}
  \end{equation}
  Using $1+r_N \le 2$ and $\sum_{k=n+1}^N k^{-1}r_N^{2N-2k} \le \sum_{k=n+1}^N k^{-1} \le 2\log N$,
  we deduce from the inequality $|A|\le |B|$ that
  \[
    \delta_n \left( C(1-r_N)r_N^{-2N}\log N + O(\frac{1}{n}) \right) \ge \left(\frac{c\pi}{2}\right)^2>0.
  \]
  With both $N\ge n\to\infty$, it is seen that $\delta_n\to 0^+$ and hence
  \begin{equation} \label{eq:ineq1}
    (1-r_N)r_N^{-2N}\log N \to \infty \quad \text{as} \quad N\to \infty.
  \end{equation}

  We then claim that the sequence $\{r_N^N\log N\}$ is uniformly bounded. Otherwise there exists an unbounded subsequence $\{r_T^T\log T \}$ (indexed with $T$) which must satisfy $r_T^T \ge (\log T)^{-1}$ for large $T$. On the other hand, Eq.(\ref{eq:ineq1}) implies that for large $T$, we have $1-r_T \ge r_T^{2T} (\log T)^{-1} \ge (\log T)^{-3}$. This leads to
  \[
  \begin{aligned}
    \log(r_T^T\log T) &\le T \log(1-(\log T)^{-3}) + \log \log T \\
    &\le -T(\log T)^{-3} + \log \log T \to -\infty \quad \text{as}
    \quad T \to \infty,
  \end{aligned}
  \]
  a contradiction to the unboundedness of $\{r_T^T\log T\}$.

  Now the estimation in Eq.(\ref{eq:Bestim}) can be refined by computing
  \[
    \sum_{k=1}^N \frac{r_N^{2N-2k}}{k} = \left( \sum_{k=1}^{\lfloor \frac{N}{2}\rfloor} + \sum_{k=\lfloor \frac{N}{2} \rfloor+1}^N \right) \frac{r_N^{2N-2k}}{k} \le 2(r_N^N \log N + \log 2) < \infty.
  \]
  As a result, the inequality $|A|\le |B|$ gives
  \[
    \delta_n \left( C(1-r_N)r_N^{-2N} + O(\frac{1}{n}) \right) \ge \left(\frac{c\pi}{2}\right)^2>0.
  \]
  By letting $N\ge n\to \infty$, we thus derive the desired estimate for $r_N$ in Eq.(\ref{eq:rseqvan}).

  (\textbf{ii}).
  If $f$ is Lipschitz, $f'$ is uniformly bounded and $f'\in L^2(-\pi,\pi)$. According to Eq.(\ref{eq:normineq}), it suffices to show, for $p'\in(1,2]$,
  \begin{equation} \label{eq:rlp}
    \sum_{k>N}r_N^{kp'} = \frac{r_N^{p'(N+1)}}{1-r_N^{p'}} \le C \frac{r_N^{Np'}}{1-r_N} \to 0 \quad \text{as} \quad N\to\infty.
  \end{equation}
  Similarly, we only need to deal with the case $\lim r_N=1$.

  Our strategy is to establish an estimation related to the lower and upper bounds of the integration $\int |\partial_r q_N(r;a)| dr$. Here $q_N(r;a)$ is defined in Eq.(\ref{eq:qnra}) for $a\in\mathbb C^{N+1}$ and $\|a\|=1$. We compute
  \[
    q_N(1;a) = \frac{1}{2\pi}\int_{-\pi}^\pi |a(\theta)|^2 \left( \sum_{|k|\ge N} \overline{m_k}e^{ik\theta} \right) d\theta
    =\frac{1}{2\pi}\int_{-\pi}^\pi |a(\theta)|^2 f(\theta)d\theta \ge c,
  \]
  where the second equality follows from the orthogonality of the bases $e^{ik\theta}$.
  Specify $a$ such that $q_N(r_N;a)=0$, and then we have
  \begin{equation}\label{eq:qnlowb}
    \int_{r_N}^1 |\partial_r q_N(r;a)| dr \ge |q_N(1;a)-q_N(r_N;a)| \ge c.
  \end{equation}

  The upper bound needs the detailed form of the derivative
  \[
    \partial_r q_N(r;a) = -\frac{r^{-1}}{2\pi} \int_{-\pi}^\pi |a(\theta)|^2 g(\theta) d\theta
  \]
  with
  \[
  \begin{aligned}
    g(\theta) &= \sum_{|k|\le N} |k| r^{-|k|} \overline{m_k}e^{ik\theta} = 2\sum_{k=1}^N kr^{-k} \Re(m_{-k}e^{ik\theta}) \\
    &= 2\sum_{k=1}^N kr^{-k} \int_{-\pi}^\pi f(t+\theta) \cos kt dt
    = -2 \int_{-\pi}^\pi f'(t+\theta) \left( \sum_{k=1}^N r^{-k}\sin kt \right) dt.
  \end{aligned}
  \]
  To derive an upper bound, we compute
  \[
  \begin{aligned}
    \sum_{k=1}^N r^{-k} \sin kt &= \Im \sum_{k=1}^N (r^{-1}e^{it})^k = \Im \frac{1-r^{-N}e^{iNt}}{re^{-it}-1} \\
    &= \frac{(1-r)r^{-N} \sin Nt + r^{1-N}(\sin Nt - \sin (N+1)t) + r \sin t}{1-2r \cos t + r^2}
  \end{aligned}
  \]
  and try to control each term. Here $\Im(z)$ denotes the imaginary part of $z\in\mathbb C$.
  Suppose $r>\frac{1}{2}$, which can be satisfied for large $N$. With the property $\sin (t/2) \ge t/\pi$ for $|t|\le \pi$, it holds that
  \[
    \frac{|\sin t|}{1-2r\cos t+r^2} = \frac{|\sin t|}{(1-r)^2+4r\sin^2\left( \frac{t}{2} \right)} \le \frac{|t|}{(1-r)^2+2\left( \frac{t}{\pi} \right)^2}
    \le \frac{C|t|}{(1-r)^2+t^2},
  \]
  resulting in
  \[
    \int_{-\pi}^\pi \frac{|\sin t|}{1-2r\cos t+r^2} dt \le C \int_0^{\pi} \frac{2t}{(1-r)^2+t^2}dt = C \log \left[ 1+ \left( \frac{\pi}{1-r} \right)^2 \right].
  \]

  Noticing $| \sin Nt - \sin(N+1)t| = 2| \sin \frac{t}{2} \cos(N+\frac{1}{2})t| \le 2| \sin \frac{t}{2}|$, we obtain
  \[
  \begin{aligned}
    |\partial_r q_N(r;a)| &\le \|r^{-1} g\|_{L^\infty} \le C \int_{-\pi}^\pi \left | r^{-1}\sum_{k=1}^N r^{-k} \sin kt \right | dt \\
    & \le Cr^{-N} + C(r^{-N} + 1) \log\left[ 1+ \left( \frac{\pi}{1-r} \right)^2 \right]
  \end{aligned}
  \]
  for $r$ sufficiently close to 1.
  To evaluate the integration, a technical result is needed:
  \[
  \begin{aligned}
    \int_{r_N}^1 \log \left[1+\left( \frac{\pi}{1-r} \right)^2\right] dr
    &= \int_0^{1-r_N} \log \left( 1+\frac{\pi^2}{t^2}\right) dt \\
    &= (1-r_N) \log \left( 1+\frac{\pi^2}{(1-r_N)^2} \right)
    + \int_0^{1-r_N} \frac{2\pi^2}{t^2+\pi^2} dt \\
    &\le (1-r_N) \left[ \log \left( 1+\frac{\pi^2}{(1-r_N)^2} \right) + 2\right],
  \end{aligned}
  \]
  where the integration by parts has been used.
  Therefore, by applying $1\le r^{-N}\le r_N^{-N}$ for $r_N\le r\le 1$, we get
  \begin{equation} \label{eq:qnupb}
  \begin{split}
    c \le \int_{r_N}^1 |\partial_r q_N(r;a)| dr &\le C r_N^{-N} \int_{r_N}^1 \log \left[1+\left( \frac{\pi}{1-r} \right)^2\right] dr\\
    &\le C r_N^{-N} (1-r_N) \log \left[1+\left( \frac{\pi}{1-r_N} \right)^2\right] \\
    &\le C r_N^{-N} (1-r_N)^{\frac{1}{p'}+\epsilon}
  \end{split}
  \end{equation}
  for any $0<\epsilon<1-\frac{1}{p'}$ because $\lim\limits_{x\to 0^+} x^\delta \log x = 0$ for any $\delta>0$.
  Rearranging the above inequality as $r_N^{Np'}(1-r_N)^{-1}\le C (1-r_N)^{\epsilon p'}$ gives Eq.(\ref{eq:rlp}) immediately.

  (\textbf{iii}).
  Suppose for the sake of contradiction that $f_N$ does not converge to $f$ uniformly. Then there exists a subsequence $\{r_T\}$ (indexed with $T$) satisfying
  \[
    r_T^{-T}(1-r_T)\le C.
  \]
  Otherwise we have $r_N^N / (1-r_N) \to 0$ as $N\to \infty$, which implies uniform convergence because $|f_N-f| = |\sum\limits_{|k|>N} e^{ik\theta} r_N^{|k|} O(1)| \le C \sum\limits_{k>N} r_N^k$.
  We shall reveal a contradiction from the integral $\int |\partial_r q_T(r;a)|dr$ in Eq.(\ref{eq:qnupb}).

  For this purpose, it is first seen that $r_T^{-[\alpha T]} (1-r_T) \to 0$ ($T\to\infty$) for any $0<\alpha<1$. The same argument as in Eq.(\ref{eq:qnupb}), with $q_N$ replaced by $q_{[\alpha T]}$, gives
  \[
    \int_{r_T}^1 |\partial_r q_{[\alpha T]}(r;a)|dr \le Cr_T^{-[\alpha T]}(1-r_T)^{\frac{\alpha+1}{2}}
    \le C \left( r_T^{-\frac{2\alpha}{\alpha+1}T}(1-r_T) \right)^{\frac{\alpha+1}{2}} \to 0
  \]
  as $T\to\infty$.

  Before we proceed, a technical result is needed:
  \[
    \left| \int f'\sin kt dt \right| \le \frac{C}{k},
  \]
  which can be verified by the equality
  \[
  \begin{aligned}
    \int_0^{2\pi} f'\sin ktdt &= \sum_{j=0}^{k-1} \left( \int_{\frac{2j\pi}{k}}^{\frac{(2j+1)\pi}{k}} + \int_{\frac{(2j+1)\pi}{k}}^{\frac{(2j+2)\pi}{k}} \right) f'\sin kt dt \\
    &= \int_0^{\frac{\pi}{k}} \sum_{j=0}^{k-1} \left[ f'(t+\frac{2j\pi}{k}) - f'(t+\frac{(2j+1)\pi}{k}) \right] \sin kt dt
  \end{aligned}
  \]
  and the assumption that $f'$ is of bounded variation.

  Then an estimate of $\int |\partial_r \left(q_T-q_{[\alpha T]}\right)|dr$ yields
  \[
  \begin{aligned}
    \int_{r_T}^1 |\partial_r \left(q_T-q_{[\alpha T]}\right)|dr
    &\le C \int \left| f'(t+\theta) \sum_{[\alpha T]<k\le T} r^{-k} \sin kt \right| dt \\
    &\le C \sum_{[\alpha T]<k\le T} r^{-k} \left| \int f'(t+\theta)\sin kt dt \right|
    \le C \sum_{[\alpha T]<k\le T} \frac{r^{-k}}{k} \\
    &\le C r_T^{-T} \cdot \frac{1}{\alpha T} \cdot (1-\alpha)T
    = C r_T^{-T} \frac{1-\alpha}{\alpha} \to 0
  \end{aligned}
  \]
  as $\alpha \to 1$.
  Therefore, we conclude that $\int_{r_T}^1 |\partial_r q_T(r;a)|dr$ can be arbitrarily small as $T\to \infty$ and $\alpha \to 1$, which contradicts the lower bound $c>0$ (see Eq.(\ref{eq:qnlowb})). This completes the proof.
\end{proof}

\bibliographystyle{amsplain}
\bibliography{references}

\end{document}